\documentclass[twocolumn]{IEEEtran}
\usepackage{graphicx}  % For including graphics
\usepackage{cite}

\usepackage{graphicx,upref}
\usepackage{amssymb,amsmath,amsthm,amsfonts}
\usepackage{graphicx}
\usepackage{xcolor}
\usepackage{booktabs}
\usepackage{multirow}
\newtheorem{theorem}{Theorem}

\newtheorem{example}[theorem]{Example}

\newtheorem{remark}[theorem]{Remark}

\newtheorem{definition}[theorem]{Definition}

\DeclareMathOperator{\Q}{\mathbf{Q}}
\DeclareMathOperator{\s}{\mathbf{s}}
\DeclareMathOperator{\x}{\mathbf{x}}
\DeclareMathOperator{\y}{\mathbf{y}}
\DeclareMathOperator{\w}{\mathbf{w}}
\DeclareMathOperator{\G}{\mathbf{G}}
\DeclareMathOperator{\C}{\mathbf{C}}
\DeclareMathOperator{\D}{\mathbf{D}}
\DeclareMathOperator{\J}{\mathbf{J}}

\DeclareMathOperator{\F}{\mathbf{F}}
\DeclareMathOperator{\FF}{\mathfrak{F}}
\DeclareMathOperator{\A}{\mathbf{A}}
\DeclareMathOperator{\V}{\mathbf{V}}
\usepackage{algorithm}
\usepackage{algpseudocode}

\usepackage{MnSymbol}

\usepackage{hyperref}

\begin{document}
\title{Graph Fourier Transform Enhancement through Envelope Extensions
 }
%\thanks{This research is supported by Montenegrin Academy of Sciences and Arts (MASA -- CANU).} }
\author{Ali Bagheri Bardi, Taher Yazdanpanah, Milo\v{s} Dakovi\'{c}, Ljubi\v{s}a Stankovi\'{c}
\thanks{
Ali Bagheri Bardi (bagheri@pgu.ac.ir) is with Persian Gulf University, Bushehr, Iran, Taher Yazdanpanah (yazdanpanah@pgu.ac.ir) is with Persian Gulf University, Bushehr, Iran, Milo\v{s} Dakovi\'{c} (milos@ucg.ac.me) and Ljubi\v{s}a Stankovi\'{c} (ljubisa@ucg.ac.me) are with University of Montenegro, Podgorica, Montenegro. \par
This work is partially supported by the Montenegrin Ministry of Education, Science and Innovations, project grant: DPG \textquotedblleft Data Processing on Graphs\textquotedblright, and by the Montenegrin Academy of Sciences and Arts.
}
}

\maketitle

\begin{abstract}
Many real-world networks are characterized by directionality; however, the absence of an appropriate Fourier basis hinders the effective implementation of graph signal processing techniques.
 Inspired by discrete signal processing, where embedding a line digraph into a cycle digraph facilitates the powerful Discrete Fourier Transform for signal analysis, addressing the structural complexities of general digraphs can help overcome the limitations of the Graph Fourier Transform (GFT) and unlock its  potential.

The Discrete Fourier Transform (DFT) serves as a Graph Fourier Transform for both  cycle graphs and Cayley digraphs on  the finite cyclic  groups $\mathbb{Z}_N$. We propose a systematic method to identify a class of   such Cayley digraphs  that can encompass a given directed graph. By embedding  the directed graph into these Cayley digraphs and opting for envelope extensions that naturally support the Graph Fourier Transform, the GFT functionalities of these extensions can be harnessed for signal analysis.

Among the potential envelopes, optimal performance is achieved by selecting one that meets key properties. This envelope's structure closely aligns with the characteristics of the original digraph. The Graph Fourier Transform  of this envelope is reliable in terms of numerical stability, and its columns approximately form an eigenbasis for the adjacency matrix associated with the original digraph.

It is shown that the envelope extensions possess a convolution product, with their GFT fulfilling the convolution theorem. Additionally, shift-invariant graph filters (systems)  are described as the convolution operator, analogous to the classical case. This allows the utilization of  systems for signal analysis.

\end{abstract}

\begin{IEEEkeywords}
	 Graph signal processing, Graph Fourier transform,  Graph filters, Cayley digraphs
\end{IEEEkeywords}

\section{Introduction}
Graph Signal Processing (GSP) is an expanding field focused on the analysis and manipulation of data residing on graphs. GSP finds diverse applications in areas such as social network analysis, sensor networks, biological networks, recommendation systems, image and video processing, brain network analysis, transportation networks, and communication systems. 
\cite{seifert2021digraph,shuman2013emerging,sandryhaila2013discrete,ortega2018graph,
sandryhaila2014big,dong2020graph,graph2020graph, isufi2024graph,li2021graph,Sig_Graphs2023,Graphs_1,trends1,
trends2, shi2019graph,huang2018graph,chen2020simple,LS_Graph}. By leveraging the underlying graph structure, GSP enables more effective data analysis and signal processing, leading to improved performance and insights across these domains.

The principal challenge in Graph Signal Processing revolves around the development of methodologies to effectively analyze and process signals within the context of graph structures. The core objective of this discourse is proposing an approach to explore and address this central challenge on directed graphs, with a primary emphasis on devising strategies that harness and exploit the inherent properties and characteristics of the underlying graph structure. This involves developing algorithms that can accommodate the directionality of edges, which are critical in many applications such as traffic flow analysis, social network dynamics, and information propagation.

Early methods in GSP often relied on spectral graph theory, which deals with the  eigen-decomposition techniques from linear algebra  \cite{domingos2020graph,sevi2023harmonic,sardellitti2017graph, seifert2023causal,seifert2021digraph,shafipour2018digraph,furutani2020graph,yang2021graph}. The key idea is to use the eigenvectors  of graph Laplacians or adjacency matrices to define transformations and operations on graph signals. Spectral methods enable tasks such as filtering, clustering, and classification by exploiting the spectral properties of the graph.

Despite their theoretical appeal, spectral methods face significant limitations when applied to directed graphs. One major challenge stems from the non-diagonalizability of adjacency matrices in many directed graphs, which complicates their eigen value decomposition for spectral analysis. Unlike undirected graphs, the adjacency matrices of directed graphs frequently yield non trivial Jordan blocks in their Jordan normal form. In addition of numerical instability, this characteristic complicates the application of shift-invariant filters, a challenge well articulated in 
 \cite{sandryhaila2013discrete,sandryhaila2014discrete}. The limitation is even more significant in large-scale graphs or in dynamic graph settings where the structure of the graph changes over time.

To ensure effective progression in our discourse, we  establish a convention to guide the discussion. Directed graphs featuring diagonalizable adjacency matrices with distinct non-zero eigenvalues will be termed admissible digraphs. 

The main contribution of this work is the establishment of a systematic approach to propose admissible options for a given directed graph, thereby enabling the application of the Graph Fourier Transform (GFT) for efficient processing of associated graph signals.

This paper is organized as follows. To facilitate the illustration of the main idea and highlight the advantages, the Introduction section is supported and enhanced by three subsections. Part A provides a general definition and deals with  the importance of admissibility, and Part B reviews Digital Signal Processing (DSP) concepts that motivated our outline. Finally, Part C  includes a brief review of two related works most closely aligned with the current subject.

Section II focuses on the theory that supports our algorithm, offering an in-depth exploration of the underlying principles and mathematical foundations. In Section III, we apply our approach to a real-world dataset, detailing the step-by-step implementation process and highlighting key findings. Lastly, an Appendix section is included to recap the essential concepts and terminology encountered throughout the paper, facilitating a comprehensive understanding of the subject matter.
%----------------------------------------------------------
\subsection{Graph Fourier Transform : Framework}
\label{framework}
Let $\G$ be a directed graph (digraph) with size $N$. A given signal $\mathbf{s}$ on $\G$ is represented as an $N$-dimensional vector in $\mathbb{R}^N$ (or $\mathbb{C}^N$). Theoretically, based on any (orthonormal) basis $\mathcal{B} = \{\mathbf{v}_0, \ldots, \mathbf{v}_{N-1}\}$, the processing of the signal $\mathbf{s}$ involves finding the scalars $y[k]$ to represent $\mathbf{s}$ in terms of the vectors in $\mathcal{B}$:
\begin{equation}
\label{eq1}
 \mathbf{s}=y[0]\mathbf{v}_0+\cdots+y[{N-1}]\mathbf{v}_{N-1}   
\end{equation}
Let $\V$ be the matrix whose columns are the vectors in $\mathcal{B}$.
Under these conditions, the representation of the signal $\mathbf{s}$ in terms of the basis vectors, as described in equation (\ref{eq1}), can be equivalently expressed using matrix notation:
\begin{equation}
\begin{cases}
\mathbf{y} = \V^{-1} \mathbf{s}, \\
\mathbf{y} = \left(y[0], \ldots, y[N-1]\right)^T.
\end{cases}
\end{equation}
When working with real datasets organized on directed graphs, various challenges can arise. To address these challenges, it is crucial to develop specific mathematical criteria that guide the analysis and processing of the data. If these criteria can be expressed in terms of an (orthonormal) basis \(\mathcal{B}\), graph signal processing techniques can be effectively applied, provided the following key conditions are met:

{\it 1) Graph structure compatibility.} The basis \(\mathcal{B}\) serves as a set of principal graph signals,  referred to as harmonics, which act as the primary components for representing signals on the graph. When processing a graph signal \(\mathbf{s}\), one crucial step involves decomposing it into these principal components to analyze how each basis vector influences the signal's properties. This decomposition process underscores the importance of establishing a connection between the basis \(\mathcal{B}\) and the underlying directed graph structure's characteristics.

{\it 2) Efficiency. } Coefficients $y[k]$s denote the coordinates of a graph signal when expressed in terms of the principal components constituting the graph Fourier basis. Manipulating the initial graph signal generally involves two primary steps: altering the Fourier coefficients and reconstructing the modified graph signal.

The adjacency matrix functor serves as a mapping between the category of directed graphs and the category of square matrices, encapsulating the connections between the graph's vertices.

Given the adjacency matrix \(\A\) associated with the digraph \(\G\), the Jordan normal decomposition allows us to factorize \(\A\) into the product of a Jordan matrix \(\mathbf{J}\) and a similarity transformation matrix \(\mathbf{V}\). This can be represented as:
\begin{equation}
\A = \V \J \V^{-1}
\end{equation}
Building upon this formula, we  define the graph Fourier transform \(\mathbf{F}\) associated with $\G$ as the inverse of the similarity transformation matrix $\mathbf{F} = \mathbf{V}^{-1}$. The columns of matrix \(\mathbf{V}\) serve as the basis generalized eigenvectors of the adjacency matrix \(\mathbf{A}\). These generalized eigenvectors constitute the graph Fourier basis \(\mathcal{B}_{\F}\), providing a set of principal graph signals on the directed graph \(\G\). 

While the graph Fourier basis \(\mathcal{B}_{\F}\),   offers valuable insights by encapsulating the principal components that capture the essential characteristics of the digraph \(\G\), there are two significant challenges that limit its widespread adoption in the realm of general graph signal processing:

(i) The first issue associated with employing the Jordan normal decomposition of the adjacency matrix 
$\A$ is its computational complexity. The decomposition process can be resource-intensive, particularly when dealing with large-scale graphs.

(ii) When the Jordan matrix is not a diagonal matrix, the utilization of {\it shift-invariant graph filters}, also known as {\it systems}, becomes impractical. These types of filters are known as essential tools for modifying the signals.
To further elaborate on this item, we reiterate that
\begin{itemize}
    \item A graph filter \(\mathbf{H}\) is shift-invariant if \(\mathbf{H} \mathbf{A} = \mathbf{A} \mathbf{H}\).
    \item For a given polynomial \(h\), the linear transformation \(h(\mathbf{A})\) is termed a system.
\end{itemize}
When the minimal and characteristic polynomials of \(\mathbf{A}\) are identical, Theorem 1 of \cite{sandryhaila2013discrete} provides a characterization for shift-invariant graph filters as follows:
\[
\text{Shift-invariant filters} = \{h(\mathbf{A}) : h \text{ is a polynomial}\}
\]

%The formulation of low, high, and band-pass graph filters requires identifying the polynomial \( h \). 

%In Section \ref{Sec:ASP}, a detailed explanation is provided, demonstrating that the theoretical formulation of shift-invariant filters, or systems, in the context of the convolution product is highly dependent on the structure of the directed graph. It is shown that this formulation holds for admissible digraphs and emphasizes the importance of employing  admissible digraphs.

In the next section, the motivation for addressing both of the mentioned challenges is provided in detail.
\subsection{Motivation : Graph Fourier Basis in  classical  DSP}
In discrete signal processing, signals are initially analyzed on a line graph where the adjacency matrix corresponds to the backward shift matrix.
\begin{equation}
\mathbf{S}=\begin{bmatrix}
0 & \mathbf{I}_{N-1}\\
0 & 0
\end{bmatrix}
\end{equation}
Here, $\mathbf{I}_{N-1}$ represents the identity matrix of size $N-1$ and  $\mathbf{S}$ denotes the backward shift matrix of size $N$.

Backward shift matrices are nil-potent, meaning all its eigenvalues are zero. While this simplifies the Jordan decomposition, the resulting graph Fourier basis does not offer substantial practical utility in signal processing tasks.

When addressing one-dimensional digital signal processing involving periodic signals,  line digraphs are replaced  with cycle digraphs. This substitution makes the backward shift matrix circulant, which is a specific type of matrix with particular properties well-suited for handling periodic signals.

\begin{equation}
\mathbf{C}=\begin{bmatrix}
0 & \mathbf{I}_{N-1}\\
1 & 0
\end{bmatrix}
\end{equation}
Here, $\C$ represents the circulant  matrix of size $N$.

 In this circular domain, the Fourier basis becomes the most organized and widely used basis for signal processing tasks. Central to this framework is the normalized discrete Fourier transform, denoted by 
$\mathfrak{F}$, defined as follows:
\begin{equation}
\mathfrak{F} = \frac{1}{\sqrt{N}} \left[ \exp(-\frac{2\pi i kl}{N}) \right]_{k,l=0}^{N-1}.
\end{equation}
The discrete Fourier transform   converts the convolution product into a pointwise operation, thereby streamlining the signal processing procedure by enabling convenient modification of the coefficients. Notably, the inverse DFT can be readily derived through its hermitian, which inherently underscores its numerical stability. 

Any periodic signal can be conceptualized as a graph signal on the directed cycle graph, as depicted in Figure \ref{cycle}. The eigenvalue decomposition of the circulant matrix, which is the adjacency matrix, is derived using the discrete Fourier transform, thereby demonstrating the intrinsic connection between the DFT and the underlying graph structure upon which periodic signals are modeled.
\begin{equation}
\begin{cases}
\C=\V
\D\V^{-1}\\
\V=\FF^{-1}\\
\mathbf{D}=\text{diag}\left(1,\exp{\frac{2\pi i}{N}},\ldots,\exp{\frac{2\pi i(N-1)}{N}}\right) \end{cases}    
\end{equation}
\begin{figure}[tbp]
\centering
\includegraphics[width=0.175\textwidth,height=0.1\textheight]{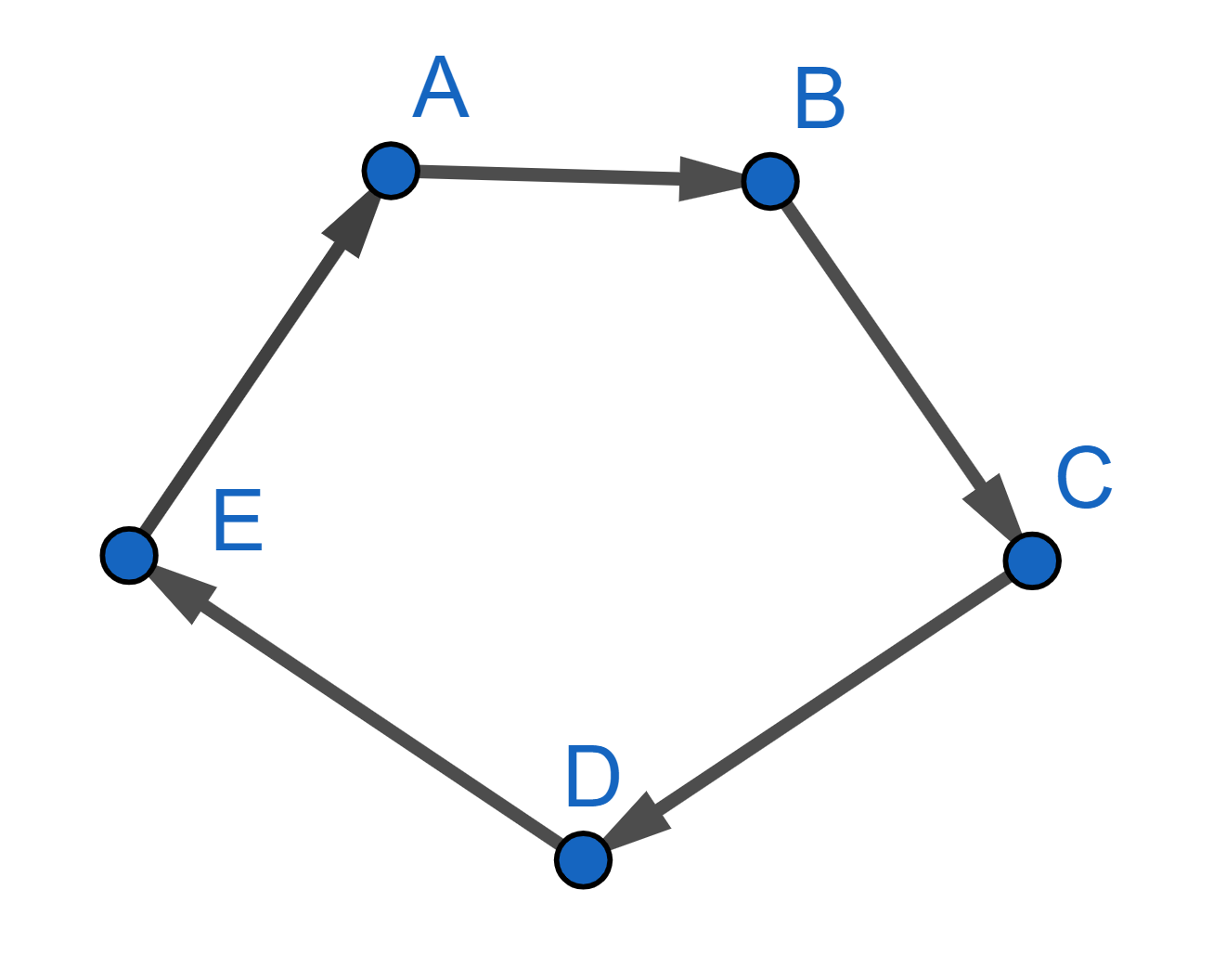}
\caption{The directed cycle graph with 5 vertices.}
\includegraphics[width=0.175\textwidth,height=0.12\textheight]{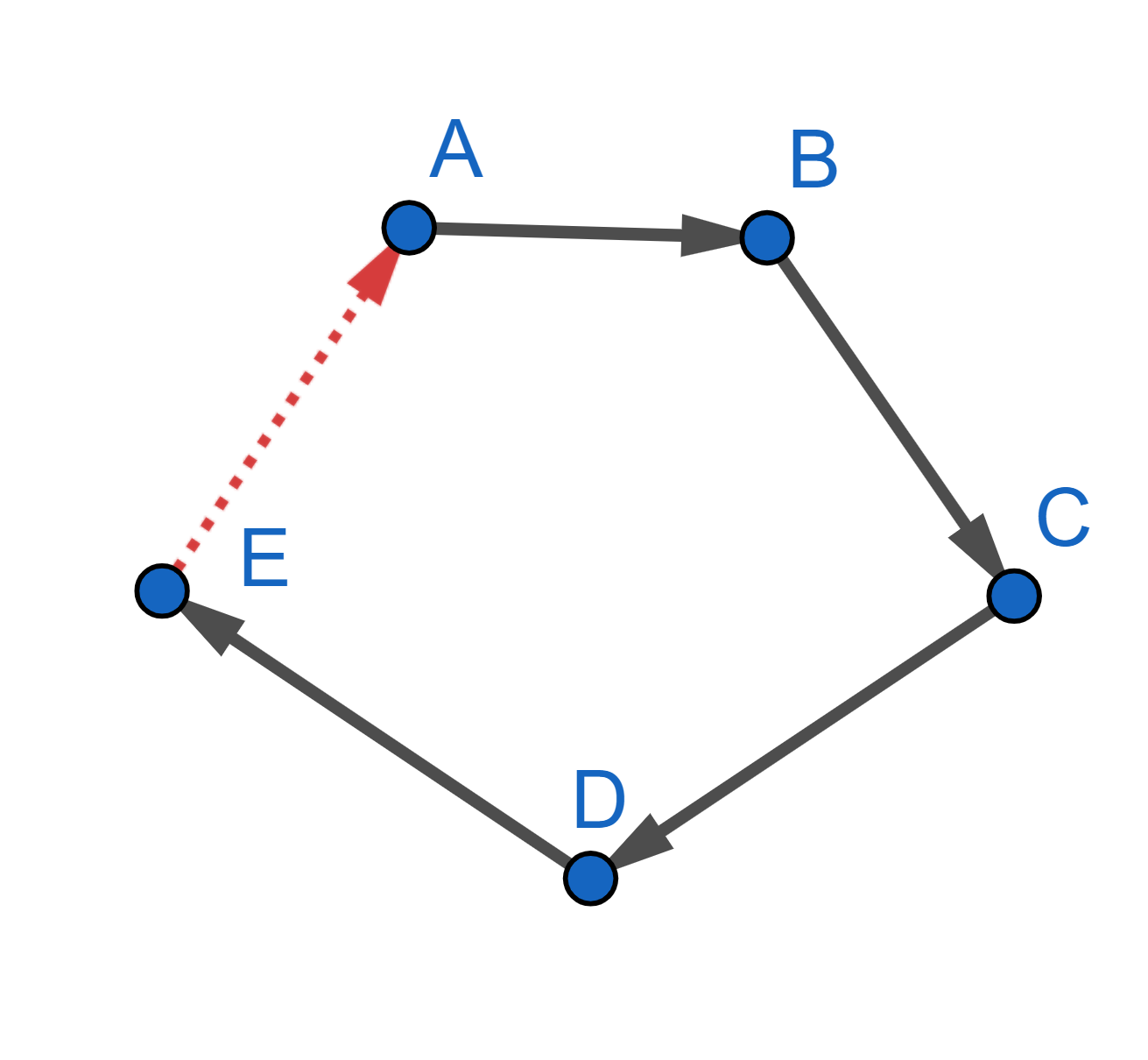}\label{cycle}
\caption{ By adding the red dotted line, the directed line graph with 5 vertices is included in the directed cycle graph.}
\label{line_cycle}
\end{figure}

In discrete signal processing, the foundational approach involves embedding the directed line graph into a directed cycle graph (see Figure \ref{line_cycle}). This transformation is achieved by connecting the sink of the line graph to its source, effectively altering a zero entry in the forward shift matrix and resulting in a circulant matrix structure. This modification turns the initial linear signal representation into a periodic signal on a circular domain.
By embedding the line graph into a cycle graph, DSP practitioners leverage the power of the discrete Fourier transform  for efficient signal processing. 

The primary contribution of this work is the development of an approach applicable to general digraphs.

\begin{definition}
\label{defn2}
Let \(\G\) be a digraph and let \(\mathbf{A}\) be the adjacency matrix associated with \(\G\) with the Jordan normal decomposition \(\mathbf{A} = \mathbf{VJV}^{-1}\).
\begin{enumerate}
    \item The digraph \(\G\) is non-singular if \(\mathbf{J}\) is a diagonal matrix with all non-zero diagonal entries.
    \item The non-singular digraph \(\G\) is admissible if the diagonal entries of \(\mathbf{J}\) are distinct.
\end{enumerate}
\end{definition}

Section \ref{Sec:ASP} establishes that admissible digraphs support a convolution product, denoted as \(\ostar\). This convolution product enables the characterization of any shift-invariant filter as a convolution operator: for a given polynomial \(h\), there exists a graph signal \(\mathbf{s}\) such that for every signal \(\mathbf{x}\), the following relationship holds:
\begin{equation}
h(\A)\mathbf{x} = \mathbf{s} \ostar \mathbf{x}
\end{equation}
Furthermore, the graph Fourier transform \(\mathbf{F}=\V^{-1}\) satisfies the convolution theorem, which is stated as follows:
 \begin{equation}
\F(\x\ostar\y)=\F\x\cdot\F\y
 \end{equation}
 In this equation, \(\mathbf{x}\) and \(\mathbf{y}\) represent graph signals, while \(\cdot\) denotes the pointwise product between the transformed vectors \(\mathbf{F}\mathbf{x}\) and \(\mathbf{F}\mathbf{y}\).

Therefore  admissible digraphs satisfy the minimal conditions necessary to employ GFT techniques for graph signal processing within the broader category of digraphs. This observation raises a question:

{\it For a given directed graph $\G$, can we devise a systematic method to construct an (optimal) admissible  directed graph $\G_e$ that includes  $\G$ as a directed subgraph?} 

Identifying such an approach could have significant implications for the efficient processing of graph signals on $\G$. This would allow us to exploit the intrinsic characteristics of admissible graph and capitalize on the potential of the GFT, which emerge from the eigenvalue decomposition of the adjacency matrix associated with $\G_e$ given in Definition \ref{defn2}.

{\it The primary contribution of this research lies in addressing the aforementioned question by effectively embedding any directed graph into a Cayley graph based on a cyclic finite group $\mathbb{Z}_N$, with the inclusion of a minimal number of new edges. These types of Cayley graphs are recognized as enriched admissible directed graphs, as explained in Theorem \ref{admissible_cayley}. It outlines a method for obtaining the optimal graph Fourier transform  with respect to numerical stability. 
}  

\subsection{Related works} The concept of incorporating new edges was first explored in \cite{seifert2021digraph}, where a method was proposed to  eliminate non-trivial Jordan blocks present in the adjacency matrix during each iteration. This process relies on the computation of the Jordan normal form  to dismantle the blocks. The method proceeds iteratively to ensure the complete removal of all non-trivial Jordan blocks. However, a major challenge arises due to the potential formation of new Jordan blocks when one is eliminated, leading to a time-consuming and prolonged process that can become difficult to control.

In the study \cite{domingos2020graph}, an intriguing attempt is proposed to define a class of graph Fourier transforms where the columns can be regarded as an approximate Fourier basis on graph $\G$. The aim is to identify some graph Fourier transforms $\F$ by solving the following optimization problem:
\begin{equation}
\label{mora}
\begin{cases}
\text{min}_{\F,\mathbf{T}}\|\A\F-\F\Lambda\| \\
\text{subject to } ~~~\A\F=\F\mathbf{T}  \\
\mathbf{T}~ \text{ is an uppertirangular matrix with zero diagonal entries}\\
~\sigma_{\text{min}}\F\geq\alpha \text{ and }  0<\alpha\leq1
\end{cases}
\end{equation}
The proposed method has demonstrated its ability to ensure spectral compatibility with the corresponding adjacency matrix. However, its effectiveness hinges on the preservation of structural properties. To elaborate this, let $\mathbf{F}$ denote the matrix obtained from the optimization problem (\ref{mora}), and let $\mathfrak{G}_F$ represent the set of all directed graphs for which $\mathbf{F}$ functions as a GFT. The greater the disparity between the graph structures in $\mathfrak{G}_F$ and the original digraph $\G$, the less reliable the results become. There is currently no evidence to suggest that this method can provide a satisfactory solution that ensures both spectral and structural compatibility.

\section{Theory} 
This section thoroughly presents the analytical evidence required to justify the methodologies employed in this study.

%First, we elucidate the fundamental theoretical underpinnings that serve as the basis for Algorithm \ref{alg-dag-connecting}, followed by the introduction of the subsequent definitions.

%For a given directed graph $\G$ , the main contribution is forming an orthonormal basis 

{\subsection{Admissible Extensions}} The primary focus of this section is to highlight that for a directed graph $\G$ with $N$ vertices, embedding it within a specific type of Cayley digraph \(\text{Cay}(\mathbb{Z}_N, \Gamma)\) can be achieved by introducing a minimal number of additional edges. This embedding ensures the existence of suitable admissible extensions of \(\G\), thereby enabling efficient graph signal processing.

The vertices of the Cayley digraph \(\text{Cay}(\mathbb{Z}_N, \Gamma)\) are given by the set \(\mathbb{Z}_N = \{0, \ldots, N-1\}\), where the connection set \(\Gamma\) determines the edges. All necessary details are provided in the appendix. We need to revisit a key point concerning the Cayley graph structure \(\text{Cay}(\mathbb{Z}_N, \Gamma)\), which will serve as a foundational framework for our subsequent analysis. The proof of this key point is also provided in the appendix.

\begin{theorem}
\label{admissible_cayley}
Suppose that $\A_{\Gamma}$ is the adjacancy matrix  of    the Cayley digraph \(\text{Cay}(\mathbb{Z}_N, \Gamma)\).
\begin{equation}\begin{cases}\A_\Gamma=\V\D_\Gamma\V^{-1} \hspace{1cm} (\V=\FF^{-1}) \\
\D_\Gamma=\sum_{q=1}^{k}\D^{n_q}\\
\D=\text{diag}\Big(1,\exp({\frac{2\pi i}{N}}),\ldots,\exp({\frac{2\pi i(N-1)}{N}})\Big) 
\end{cases}\end{equation}
\end{theorem}

Theorem \ref{admissible_cayley} demonstrates that the discrete Fourier transform  serves as a GFT for all types of these Cayley digraphs.

The following introduces the principal theorem that underpins the algorithm used to introduce the graph Fourier transform for any digraph. To ensure a comprehensive understanding of this theorem, it is beneficial to adopt the following convention:
 
{\it Dependency lists and pseudo-permutation.} Given an \(N \times N\) matrix \(\mathbf{A}\) with a nullity index \(g\), representing the dimension of its null space, we introduce the concepts of column and row dependency lists, denoted as \(L_{\text{DepCol}}\) and \(L_{\text{DepRow}}\), respectively. A row dependency list \(\mathcal{R}\) in \(L_{\text{DepRow}}\) consists of \(g\) rows of matrix \(\mathbf{A}\), chosen such that removing the rows indicated in \(\mathcal{R}\) does not affect the rank of \(\mathbf{A}\). Similarly, a column dependency list \(\mathcal{C}\) in \(L_{\text{DepCol}}\) comprises \(g\) columns of \(\mathbf{A}\) with the same property.
For a given pair \((\mathcal{R}, \mathcal{C})\), we define an \(N \times N\) matrix \(\mathbf{Q}\) with rank \(g\) and unit non-zero entries located exclusively at the intersections of \(\mathcal{R}\) and \(\mathcal{C}\) as a pseudo-permutation. Notably, there are \(g!\) distinct pseudo-permutations corresponding to each pair \((\mathcal{R}, \mathcal{C})\).

\begin{figure}[tbp]
\centering
\includegraphics[width=0.23\textwidth,height=0.125\textheight]{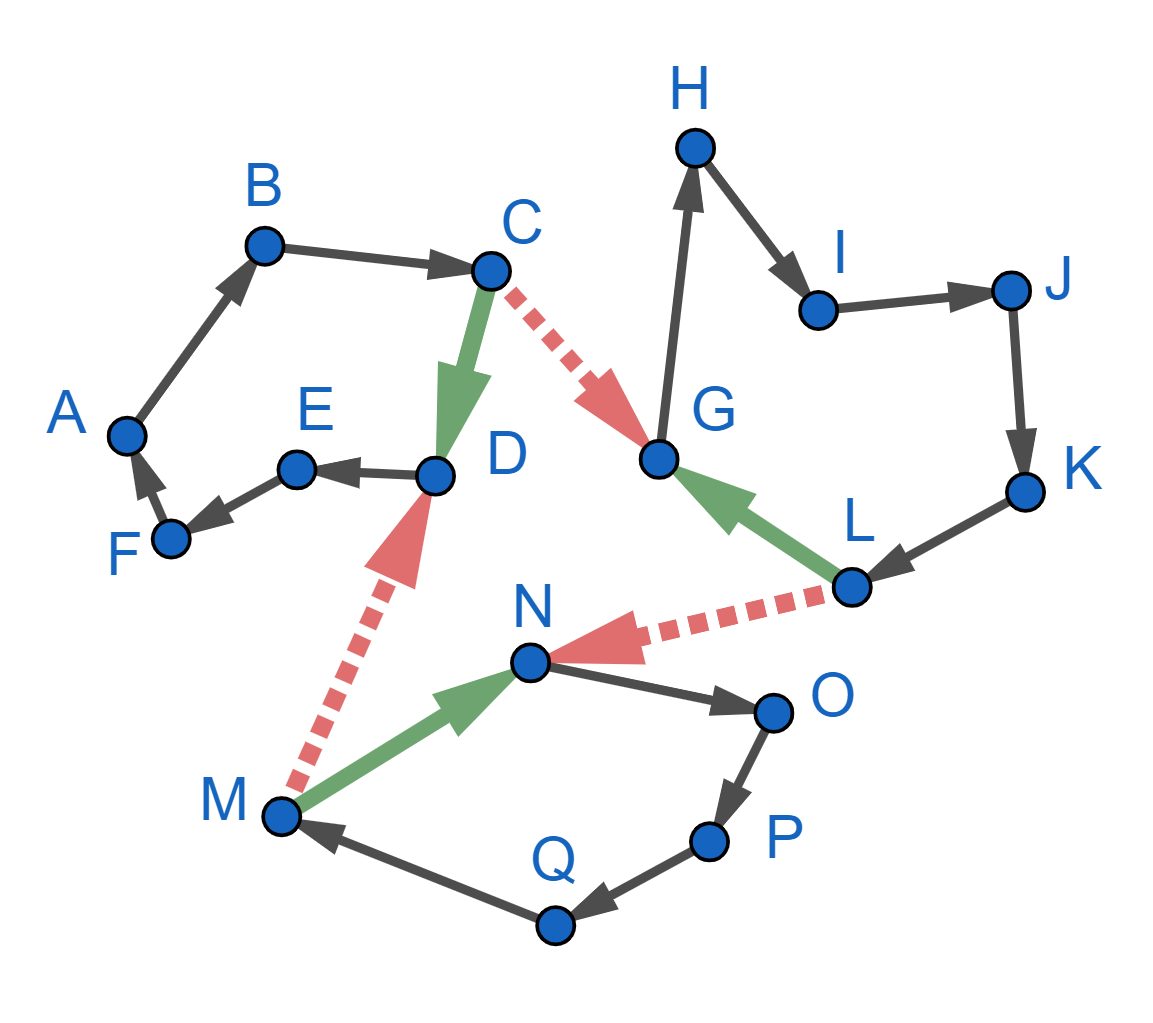}
\caption{Hamiltonian Cycle Creation}
\label{disjoint_cycles}
\end{figure}
\begin{theorem}\label{Edge}
Let \(\G\) be a digraph with \(N\) vertices. There is a Cayley graph \(\text{Cay}(\mathbb{Z}_N, \Gamma)\) that includes \(\G\) as a directed subgraph.
\end{theorem}
\begin{proof}
Let \(\A\) denote the adjacency matrix associated with \(\G\). The proof is conducted in three steps.

{\it Step1. Non-singularization:} If $\A$ is singular, this step must be performed; otherwise, it can be skipped. Suppose that the nulity index of $\A$ is $g$. Let $\mathcal{C}$  (resp. $\mathcal{R}$)  be  $g$ columns (resp. $g$ rows)  representing a  column dependency list (row dependency list) of $\A$ and  $\mathbf{Q}$ be a corresponding pseudo-permutation to the pair ($\mathcal{R}$,$\mathcal{C}$). We  define $\A_{\text{Inv}}=\A+\mathbf{Q}$.  

By excluding the columns $\mathcal{C}$ and rows $\mathcal{R}$ from $A$, we obtain the   $(N-g)\times(N-g)$  submatrix of $\A$, denoted as $\A_0$ that is definitely  non-singular. Using the Laplace expansion to compute the determinant of $\A_{\text{Inv}}$ and starting from either the columns $\mathcal{C}$ or the rows  $\mathcal{R}$, we directly observe that 
\[|\det(\A_{\text{Inv}})|=|\det(\A_0)|\]
Let $\G_{\text{Inv}}$ be the digraph whose adjacancy matrix is  $\A_{\text{Inv}}$. It is indeed created by adding $g$ new edges to 
$\G$, as specified by the pseudo-permutation $\mathbf{Q}$.  

 {\it Step2. Hamiltonian Cycle Creation:} We represent, 
 \begin{equation}
 \A_{\text{Inv}}=
\begin{bmatrix}
a_{ij}
\end{bmatrix}_{i,j=0}^{N-1}
 \end{equation}
  Employing the definition of the determinant of $\A_{\text{Inv}}$ as a sum over permutations,
\begin{equation}
|\det(\A_{\text{Inv}})|=\sum_{\sigma}\prod_{0\leq i\leq N-1} a_{i,\sigma(i)}    
\end{equation}
where the sum is taken over all permutations $\sigma$ on the set $\{0,\ldots,N-1\}$. Because $\A_{\text{Inv}}$ is non-singular, there exists a permutation $\sigma$ such that $\prod_{0\leq i\leq N-1} a_{i,\sigma(i)}$ is non-zero. This $\sigma$ corresponds to a collection of $k$ disjoint cycles $\mathcal{O}_1,\ldots, \mathcal{O}_k$ that cover all vertices of $\G_{\text{Inv}}$. As visualized in Figure \ref{disjoint_cycles} for \(k=3\), by adding at most \(k\) new edges, it is possible to connect the cycles \(\mathcal{O}_1, \ldots, \mathcal{O}_k\) to form a cycle \(\mathcal{O}\) that covers all vertices of \(\G_{\text{Inv}}\). First, select an arbitrary edge on each of these cycles (green ones). Then, using the new edges (red ones), connect the cycles \(\mathcal{O}_m\) to form the cycle \(\mathcal{O}\).

{\it Step3. Constructing the Cayley Graph:} 
The cycle digraph $\mathcal{O}$ is isomorphic to digraph $\mathbb{Z}_N$. Let $E$ be the set of all edges present in $\G_{\text{Inv}}$, except those on the cycle $\mathcal{O}$. Note that $E$ may be considered as  subset of $\mathbb{Z}_N \times \mathbb{Z}_N$. We define:

\[
\begin{cases}
\Gamma_0 = \{|a - b| : (a, b) \in E\} \\
\Gamma = \Gamma_0 \cup \{1\}
\end{cases}
\]
It is evident that the initial digraph $\G$ is embedded into  the Cayley digraph \(\text{Cay}(\mathbb{Z}_N, \Gamma)\). 
\end{proof}

\begin{remark}
Upon examining the proof of Theorem \ref{Edge}, it becomes evident that the assignment of specific non-zero weights to the additional edges does not  affect the attainment of non-singularity. Instead, the existence of these edges plays a pivotal role in ensuring the adjacency matrix's non-singularity. This implies that, as long as the newly introduced edges are included in the adjacency matrix with any non-zero weight, the resultant matrix will inherently possess the property of non-singularity. During this study, by an \textit{admissible extension} we mean the newly added edges are considered with no weight unless specifically pointed out.
% This observation implies that the focus should be on identifying and adding the necessary edges to the graph structure to ensure a non-singular adjacency matrix, rather than on the particular weights assigned to these edges.
\end{remark}

In light of Theorem \ref{Edge}, we find it suitable to present the following definition as a foundation for our continued discussion:

\begin{definition}
Let \(\G\) be a digraph. Any admissible (non-singular) digraph that contains \(\G_e\) as a subgraph is referred to as an admissible extension (non-singular extension) of \(\G\).
\end{definition}

Several considerations deserve attention with regard to Theorem \ref{Edge}.\medskip

{\bf (i)} The challenge of embedding a digraph within Cayley graphs is of great interest in graph theory, a field concentrating on the analysis and manipulation of graph structures \cite{godsil,przezdziecki}. Theorem \ref{Edge} serves as a solution to this important problem, offering a straightforward tool for researchers in this area.

{\bf (ii)}  When the nullity index \(g\) of the adjacency matrix associated with the directed graph is greater than one, several  non-isomorphic non-singular extensions will emerge. Let \(\beta_{\text{Dep\_Col}}\) and \(\beta_{\text{Dep\_Row}}\) denote the number of different dependency column and row lists, respectively, and define 
\begin{equation}
\label{total}
 \text{Total}_{\text{Inv}}(\G) = \beta_{\text{Dep\_Col}} \times \beta_{\text{Dep\_Row}} \times g!
\end{equation}
Theorem \ref{Edge} provides numerous \(\text{Total}_{\text{Inv}}\) non-singular extensions of the digraph \(\G\). This leads to the potential existence of many non-isomorphic admissible extensions.

${\bf (ii)}$  {Drawing from experience, it is widely  observed that on any directed graph, the non-singular extension is typically just an admissible extension as well}.
 However, in cases where an admissible  extension does not immediately materialize after forming the non-singular extension, the corresponding Cayley graph extension identifies the positions where additional edges could be inserted to achieve the desired outcome. In these scenarios, the admissible extension  is often achieved after just one or two more edges are added.  These observations  serve as a guide for achieving the  admissible extension and  provide theoretical support for the existence of  admissible extensions for any arbitrary directed graph.

{\bf (iii)} The methodology, as delineated in Theorem \ref{Edge}, comprises three distinct adjustments made to both the directed graph and its corresponding adjacency matrix through the introduction of new edges.  Initially, the focus is on ensuring non-singularity, followed by the establishment of a Hamiltonian cycle in the subsequent step. Finally, the process progresses towards achieving diagonalizability. 

The inclusion of new edges in the first phase precisely corresponds to the nullity of the adjacency matrix. After completing this step, all Jordan blocks associated with the eigenvalue \(0\) are eliminated.

\begin{algorithm}[tb]
    \caption{Admissible Extension Algorithm}
    \label{Enveloping Graphs}
    \begin{algorithmic}
        \Require Adjacency matrix $\mathbf{A}$ of directed graph $\G$
        \State $N \gets$ size($\mathbf{A}$) \Comment{Number of nodes in graph $\G$}
        \State $r \gets$ rank($\mathbf{A}$) \Comment{Rank of matrix $\mathbf{A}$}
        \State $g \gets N - r$ \Comment{Nullity of matrix $\mathbf{A}$}
        \State Find all lists $\mathcal{R}$ of $r$ linearly independent rows of $\mathbf{A}$
        \State Find all lists $\mathcal{C}$ of $r$ linearly independent columns of $\mathbf{A}$
        \State For each pair $(\mathcal{R}, \mathcal{C})$, construct the $N \times N$ pseudo-permutation matrix $\mathbf{Q}$ using $g$ rows not in $\mathcal{R}$ and $g$ columns not in $\mathcal{C}$
        \State $\mathbf{A}_{\text{inv}} \gets \mathbf{A} + \mathbf{Q}$ \Comment{Modified adjacency matrix}
        \State Verify all possible choices of $(\mathcal{R}, \mathcal{C})$ to ensure the admissibility of $\G_{\text{Inv}}$, the digraph associated with $\mathbf{A}_{\text{inv}}$
        \Ensure A class of admissible digraphs  that include $\G$ as a subgraph.
    \end{algorithmic}
\end{algorithm}

%\begin{remark}
%{\color{red} If we regard $\G_e$ as an non-singular extension of $\G$, it seems  that the %characteristic and minimal polynomials of the adjacency matrix $\A_e$ coincide. %Additionally, computing the Jordan normal decomposition of $A_e$ for such matrices %proceeds smoothly. This observation suggests that an non-singular extension of a directed %graph would be suitable for the graph Fourier transform.}
%\end{remark}

\subsection{Spectral components: Admissible Extensions}
Let $\G_e$ be  an  admissible extension of $\G$ with the associated adjacency matrices $\A_e$ and  $\A$. Let  $\Q$ be a corresponded  pseud-permutation with 
\begin{equation}
\label{eq_psudo}
\A_e=\A+\Q
\end{equation}  
Let $\F=\V^{-1}$ be a GFT on $\G_e$. The columns of  $\V$ serve a basis compsed of right eigenvectors of $\A_e$.  We recall that each eigenvalue of $\A_e$ possesses an algebraic multiplicity of one, indicating that   the eigenvectors are unique up to scalar multiplication. 
\begin{equation}
\begin{cases}
\A_e\V=\V\D\\
\D=\text{Diag}(\lambda_0,\ldots,\lambda_{N-1}), \\
\lambda_k\text{s are eigenvalues of } \A_e
\end{cases}
\end{equation}
By leveraging  (\ref{eq_psudo}), we can derive the equation
\begin{equation}
\label{app}
-(\A\V-\V\D)=\Q\V
\end{equation}
This expression indicates that \(\F\), the graph Fourier transform of \(\G_e\), functions as an approximate graph Fourier transform for \(\G\), with an error term represented by \(\Q\V\). The non-zero rows of \(\Q\V\) capture the influence of the edges required to achieve the admissibility of \(\G_e\). This matrix helps determine how effectively the eigenvectors of \(\A_e\) can serve as approximations for the eigenvectors of \(\A\).

Let's consider the $k$-th column vector of $\V$, denoted as $\w_k$, which satisfies the equation 
\begin{equation}
    \A\w_k - \lambda_k\w_k = \Q\w_k
\end{equation} 
The vector \(\Q\w_k\), whose  nonzero components of \(\w_k\) are just appeared in the row indices addressed by $\Q$, helps us evaluate how well \(\w_k\) can serve as an approximate eigenvector for \(\A\). To continue with the discussion, let us consider the following  $(\delta,\Delta)$ indices:
\begin{equation}
\label{indices}
\begin{cases}
\Delta_{\V}=||\Q\V||\\
\delta_{\V}=\text{max}\{||\Q\w_k||: k=0,\ldots, N-1\}    
\end{cases}
\end{equation}
%The discussions and relations, examined in this section, contribute the following key points to the primary contributions of this paper:

%(i) Normalization will be necessary to standardize the definitions of \((\delta, \Delta)\) indices in the current discussion. The graph Fourier transform \(\F\) used is assumed to be normalized to have a unit operator norm and to have the specific property that the first non-zero component of each eigenvector \(\w_k\) is \(1\).

(i) Trivially, for any arbitrary small $\epsilon>0$ and eigenvector $\w_k$, the scaled vector $\epsilon\w_k$ is also an eigenvector of $\A_e$ associated with the eigenvalue $\lambda_k$. We may write, 
\begin{equation}
\label{epsilon}
\begin{cases}
\A_e\V_\epsilon=\V_\epsilon\D   \\
\V_\epsilon=\epsilon\V
\end{cases}
\end{equation}
Indeed $\V_\epsilon$ corresponds some graph Fourier transform  of the admissible digraph $\G_e$. The second  condition in  (\ref{epsilon}) clarifies that each column of $\V_\epsilon$ is a scaled version of the original one $\V$ by the factor $\epsilon$.
Theoretically, introducing $\V_\epsilon$ allows for assuming both indices $\delta_{V_\epsilon}$ and $\Delta_{\V_\epsilon}$ to be sufficiently small. However, in practice, the choice of $\epsilon$ poses a trade-off with the numerical stability of $\V_\epsilon$. A smaller $\epsilon$ might satisfy the assumption of small indices but could also result in a reduced numerical stability when working with $\V_\epsilon$. This trade-off necessitates careful consideration in real-world applications, as both spectral  and numerical stability are essential factors for accurate and reliable results. Striking an optimal balance between the two aspects is crucial to ensure the best performance of the graph-based system under investigation. 

(ii) In the context of a weighted admissible extension  for $\G$, there is a relationship between the  weights assigned to the newly added edges and the norm of the weighted pseudo-permutation matrix. Specifically, as the edge weights are chosen to be smaller, the norm of the weighted pseudo-permutation matrix also decreases. This direct impact is evident in the delta indices (\ref{indices}), which become smaller as well. Choosing excessively small weights might compromise the numerical stability of the graph Fourier transform, potentially leading to unreliable results.

If the primary objective is to obtain a numerically stable graph Fourier transform for a given graph $\G$ while minimizing the indices introduced in (\ref{indices}), it is advisable to explore the entire domain of weighted and unweighted admissible extensions. By examining various admissible extensions, one can identify the optimal graph that strikes the best balance between accurate eigenvector estimates and numerical stability, highlighting the benefits of this current approach.

\begin{example}
\label{ex10}
Let $\G$ be the directed line graph of size $N$ whose adjacency matrix is $\mathbf{S}$, the backward shift matrix. Connecting the sink to the source in this case, a unique rank-one pseudo-permutation $\Q$ exists, making $\mathbf{S}$ the circulant matrix $\C$. 

In the unweighted scenario, the graph Fourier transform is just $\FF$. The following values   are simply formulated:
\begin{equation}
\label{F_delta}
\begin{cases}
\Delta_{\V}=1  \hspace{1.1cm} \V=\FF^{-1} \\
\delta_{\V}=\frac{1}{\sqrt{N}}\\
{k}(\V)=1 \hspace{1cm} \text{condition number}
\end{cases}    
\end{equation}
Now suppose that the connection of the sink to the source is considered with the weight 
\( w \) and denoted as \(\C_w \) with the adjacency matrix of the weighted admissible digraph.
\begin{equation}
\C_w=\begin{bmatrix}
0 & \mathbf{I}_{N-1} \\
w & 0
\end{bmatrix}
\end{equation}
The characteristic polynomial of \( \C_w \) is just \( p(\lambda) = \lambda^N - w \). Thus, the eigenvalues of \( \C_w \) are just the  $N$-th roots of $w$ and they are given  as follows when $w$ is a positive real number:
\begin{equation}
\lambda_k = w^{1/N} \exp{\left(\dfrac{2\pi i k}{N}\right)} \quad \text{for} \quad k = 0, 1, \ldots, N-1
\end{equation}
We define,  
\begin{equation}
\label{f9}
\V_w=\dfrac{1}{\sqrt{N}}\begin{bmatrix}
1 & 1 & \cdots & 1 \\
\lambda_0 & \lambda_1 & \cdots & \lambda_{N-1} \\
\vdots &   \vdots   & \cdots &  \vdots \\
\lambda_0^{N-1}  & \lambda_1^{N-1} & \cdots  & \lambda_{N-1}^{N-1}
\end{bmatrix}
\end{equation}
One may directly check that the $k$th column of $\V_w$ is a right eigenvector of $\C_w$ corresponding to the eigenvalue $\lambda_k$. This yields:
\begin{equation}
\label{f10}
\begin{cases}
\C_w=\V_w\mathbf{D}_w\V^{-1}_w\\
\mathbf{D}_w=\text{diag }({\lambda}_0,\ldots,{\lambda}_{N-1})\\
\V_w=\mathbf{E}_w\FF^{-1}\\
\mathbf{E}_w=\text{diag }\big(1, w^{\frac1N},\ldots, w^{\frac{N-1}{N}}\big).\\
\end{cases}
\end{equation}
Notably $\V_w$ is normalized when $w\leq1$. Furthermore, 
\begin{equation}
\begin{cases}
\Delta_{\V_w}=w \medskip \\
\delta_{\V_w}~=\dfrac{w}{\sqrt{N}} \medskip\\
k(\V_w)=\dfrac{1}{w^{\frac{N-1}{N}}}
\end{cases}
\end{equation}

\begin{figure}[tbp]
\centering
\includegraphics[width=0.45\textwidth,height=0.15\textheight]{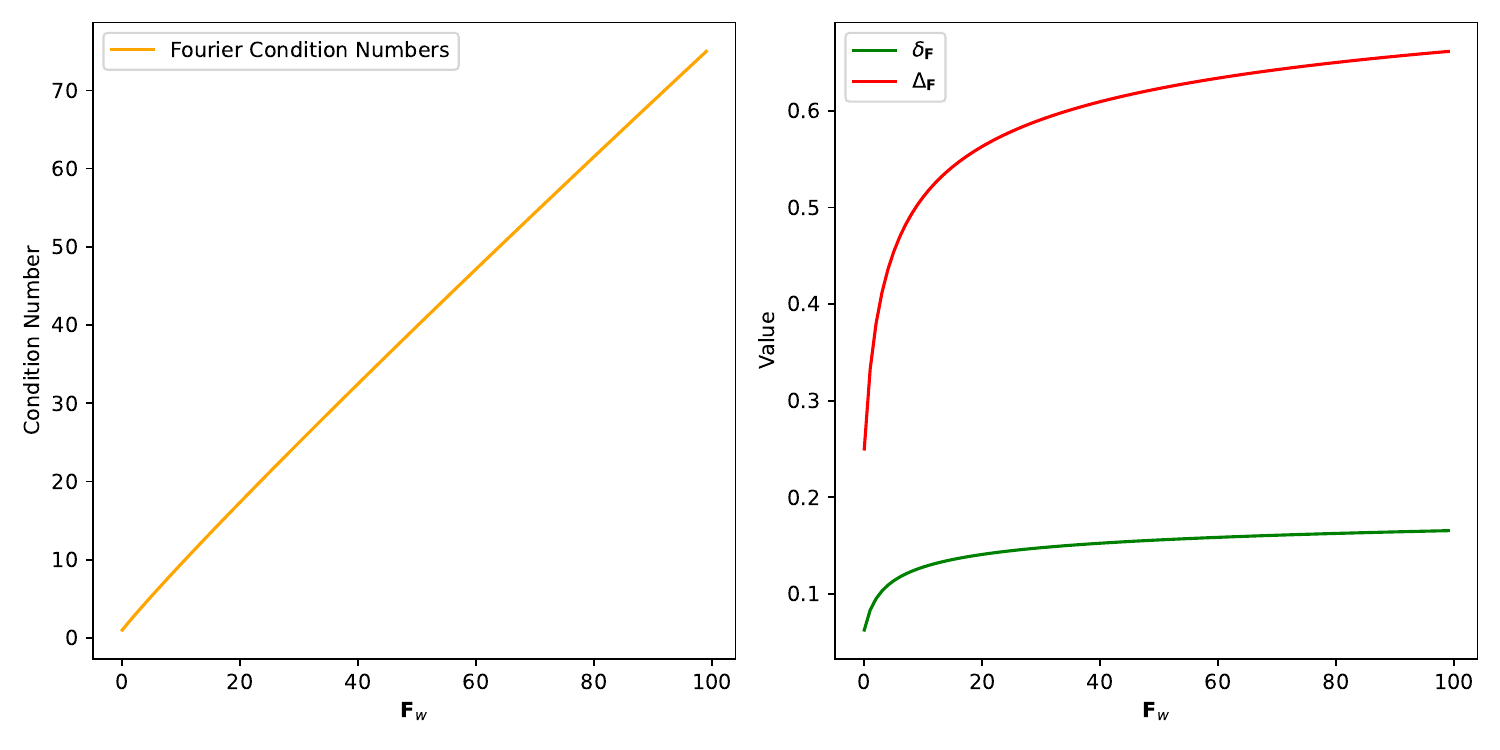}
\caption{The figure shows the $(\delta,\Delta)$ indices and  condition numbers of \(\F_w\) for weights ranging from \(1\) to \(0.01\) with a step size of \(0.01\) when \(N=16\). The minimum and maximum condition numbers are reported as \(1\) and \(74.98942\).
}
\label{DFT_ww_delta_and_Delta}
\end{figure}
 In Figure \ref{DFT_ww_delta_and_Delta}, the condition numbers and \((\delta, \Delta)\) indices of \(\F_w\) are illustrated for weights ranging from \(1\) to \(0.01\) with a step size of \(0.01\), when \(N=16\).
\end{example}

\subsection{Convolution Product: Admissible Digraphs \label{Sec:ASP}}
As mentioned in the introduction, to apply graph signal processing techniques to a digraph \(\G\), it is essential for the graph to support a convolution product capable of accommodating all possible systems.

In the domain of digital signal processing, the cyclic group \(\mathbb{Z}_N\) provides such a convolution product. For two discrete-time signals \( x[n] \) and \( y[n] \) with a period \( N \), their (normalized) cyclic convolution is defined as follows:
\begin{equation}\label{conv-dsp}
(x \ast y)[n] = \frac{1}{\sqrt{N}} \sum_{m=0}^{N-1} x[m] y[(n - m)]  
\end{equation}
Two key points make the conventional convolution product \( \ast \) efficient. First, the discrete Fourier transform \( \FF \) converts it to the point-wise product:
\begin{equation}
\FF(\x \ast \y) = \FF \x \cdot \FF \y\\
\end{equation}
Any system \( \mathbf{H} \) is formulated by \( \mathbf{H} \x = \x \ast \s \) for some unique signal \( \s \).

The class of digraphs supporting such a convolution would be reliable for utilizing graph filters for signal analysis. The subsequent definition offers a mathematical expression for this desired convolution product:
\begin{definition}
Let \(\G\) be a digraph with the associated adjacency matrix $\A$ and GFT  \(\mathbf{F}\)  on \(\G\). We say that a binary operation \(\ostar\) on the set of signals on \(\G\) defines a convolution product if the following conditions are satisfied:

\begin{enumerate}
    \item The convolution theorem holds: for given signals \(\mathbf{x}\) and \(\mathbf{y}\) on \(\G\),
    \begin{equation}
        \mathbf{F}(\mathbf{x} \ostar \mathbf{y}) = \mathbf{F}\mathbf{x} \cdot \mathbf{F}\mathbf{y}
    \end{equation}
    where \(\cdot\) denotes the pointwise product between the transformed signals \(\mathbf{F}\mathbf{x}\) and \(\mathbf{F}\mathbf{y}\).
    \item Any system is a convolution operator: for a given polynomial \(h\), there exists a signal \(\mathbf{s}\) on \(\G\) such that for an arbitrary signal \(\mathbf{x}\) on \(\G\),
    \begin{equation}
        h(\mathbf{A})\mathbf{x} = \mathbf{s} \ostar \mathbf{x}
    \end{equation}
    \end{enumerate}
\end{definition}

\begin{theorem}
\label{con-thm}
Every admissible digraph  $\G$ possesses a  convolution product. 
\end{theorem}
\begin{proof}
Let $\A$ be the associative adjacency matrix of $\G$ with the eigenvalue decomposition $\A=\F^{-1}\D\F$. 

Let \(\mathbf{1}\) be the constant signal whose all components are \(1\). For any given polynomial \(h\), we define 
\begin{equation}
\label{f7}
\mathbf{s}_h = \mathbf{F}^{-1} h(\mathbf{D}) \mathbf{1}
\end{equation}

The admissibility property allows, by leveraging the Lagrange interpolation polynomials, any graph signal \(\mathbf{x}\) on \(\G\) to be uniquely represented by \(\mathbf{x} = \mathbf{s}_{h_{\x}}\) for some polynomial \(h_{\x}\). Utilizing this key point, we define:
 
\begin{equation}
\begin{cases}
\ostar:\mathbb{C}^N\times\mathbb{C}^N\to\mathbb{C}^N \\
\x\ostar\y=\s_{h_{\x}h_{\y}}
\end{cases}
\end{equation}
Here $h_{\x}h_{\y}$ is just the product of polynomials $h_{\x}$ and $h_{\y}$. It is evident that \(\mathbf{x} \ostar \mathbf{y}=\mathbf{y} \ostar \mathbf{x}\). Using (\ref{f7}), it can be directly verified that for a set of signals \(\mathbf{x}_1, \ldots, \mathbf{x}_k\) and \(\mathbf{y}_1, \ldots, \mathbf{y}_k\), and scalars \(x_1, \ldots, x_k\) and \(y_1, \ldots, y_k\), the following property holds:
\begin{equation}
\Big(\sum_{m=1}^k a_m\x_m\Big)\ostar\Big(\sum_{n=1}^kb_n\y_n\Big)=\sum_{m,n=1}^k a_mb_n\x_m\y_n
\end{equation}
Now, we check $\ostar$ satisfies the convolution theorem. 
\begin{align*}
\F(\x\ostar\y)&=\F(\s_{h_{\x}h_{\y}})=\F(\F^{-1}h_{\x}h_{\y}(\D)\mathbf{1})\\
&=h_{\x}h_{\y}(\D)\mathbf{1}=\Big(h_{\x}(\D)\mathbf{1}\Big)\cdot\Big(h_{\y}(\D)\mathbf{1}\Big)\\
&=\Big(\F\s_{h_{\x}}\Big)\cdot\Big(\F\s_{h_{\y}}\Big)\\
&=\F\x\cdot\F\y
\end{align*}
Finally we check that the system $h(\A)$ is a convolution operator. 
\begin{align*}
h(\A)\x&=\F^{-1}h(\D)\F\x=\Big(\F^{-1}h(\D)\F\Big)\s_{h_{\x}}\\
&=\F^{-1}hh_{\x}(\D)\mathbf{1}=\s_{hh_{\x}}\\
&=\s_{h}\ostar \x
\end{align*}
\end{proof}

\begin{remark}
With the notation of Theorem \ref{con-thm}, let \(\mathbf{d}_h\) be the signal obtained by the diagonal entries of the matrix \(h(\mathbf{D})\), forming the set of its eigenvalues. As mentioned earlier, the columns of \(\mathbf{V} = \mathbf{F}^{-1}\) constitute the eigenvectors of \(h(\mathbf{D})\). One may directly see that,
\begin{equation}
\s_h=\V\mathbf{d}_h    
\end{equation}
\end{remark}

\begin{example}
\label{ex6}
Consider the line digraph \(\G\) with $N$ vertices  and its admissible extension, the cycle digraph \(\G_w\), which is formed by connecting the sink to the source with a positive weight \(w\). 
We compute the convolution product \(\ostar\) on \(\G_w\). As demonstrated in Example \ref{ex10}, \(\F_w=\V^{-1}_w\) serves as the GFT on \(\G_w\), where \(\V_w\) is defined in (\ref{f9}). For given unit impulse signals \(\delta[m]\) and \(\delta[n]\), we have:
{\small\begin{align*}
\delta[m] \ostar \delta[n]& =\F_w^{-1}\Big(\F_w\delta[m]\cdot\F_w\delta[n]\Big)\\
&=\mathbf{E}_w\FF^{-1}\Big(\FF\mathbf{E}^{-1}_w\delta[m]\cdot\FF\mathbf{E}^{-1}_w\delta[n]\Big)
\\&=\mathbf{E}_w\FF^{-1}\Bigg(\FF\Big(\mathbf{E}^{-1}_w\delta[m] \ast \mathbf{E}^{-1}_w\delta[n]\Big)\Bigg)
\\&=\mathbf{E}_w\Big(\mathbf{E}^{-1}_w\delta[m] \ast \mathbf{E}^{-1}_w\delta[n]\Big)
\\&=\mathbf{E}_w\Big(w^{-\frac{m}{N}}\delta[m] \ast w^{-\frac{n}{N}}\delta[n]\Big)
\\&=w^{-\frac{m+n}{N}}\mathbf{E}_w\Big(\delta[m] \ast \delta[n]\Big)=\sqrt{N}w^{-\frac{m+n}{N}}\mathbf{E}_w\delta[m+n]
\\&=\sqrt{N}\delta[m+n]=\delta[m] \ast \delta[n]
\end{align*}}
where \(m+n\) is taken modulo \(N\).

Since any signal can be represented as a linear combination of unit impulse signals, the convolution product \(\ostar\) on the directed graph \(\G_w\) becomes the same as the conventional convolution product \(\ast\).
\end{example}

To determine the optimal admissible extension for the line digraph with \(N\) vertices, we can combine the insights obtained from Examples \ref{ex10} and \ref{ex6}. These examples demonstrate that the cyclic digraph \(\mathbb{Z}_N\) is the most suitable admissible extension for the line digraph.

\section{Application  : Friendship Digraph} 
Upon examining the discourse leading to equation (\ref{total}), it becomes apparent that the number of admissible extensions for a digraph with a singular adjacency matrix can potentially be quite substantial. 
 This finding is further supported a by real-world data sets, given in subsequent section, which demonstrate that the actual number of admissible extensions far surpasses what one might anticipate at the outset.  Recognizing the abundance of admissible extensions, identifying the optimal choice that caters to the targeted goal based on the specific data set under examination is highly beneficial. 
\subsection{ Friendship Digraph in a High School in Illinois}
\begin{figure}[tbp]
\centering
\begin{minipage}{0.48\textwidth}
    \centering
    \includegraphics[width=0.7\textwidth,height=0.2\textheight]{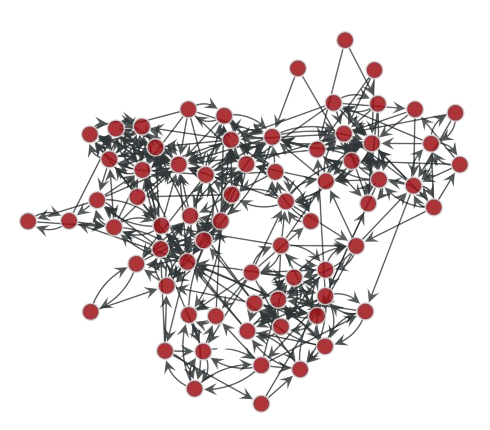}
    \caption{Friendship digraph $\G_{\text{f}}$ with 70 vertices.}
    \label{Friendship}
\end{minipage}\hfill
\begin{minipage}{0.48\textwidth}
    \centering
    \includegraphics[width=0.75\textwidth,height=0.25\textheight]{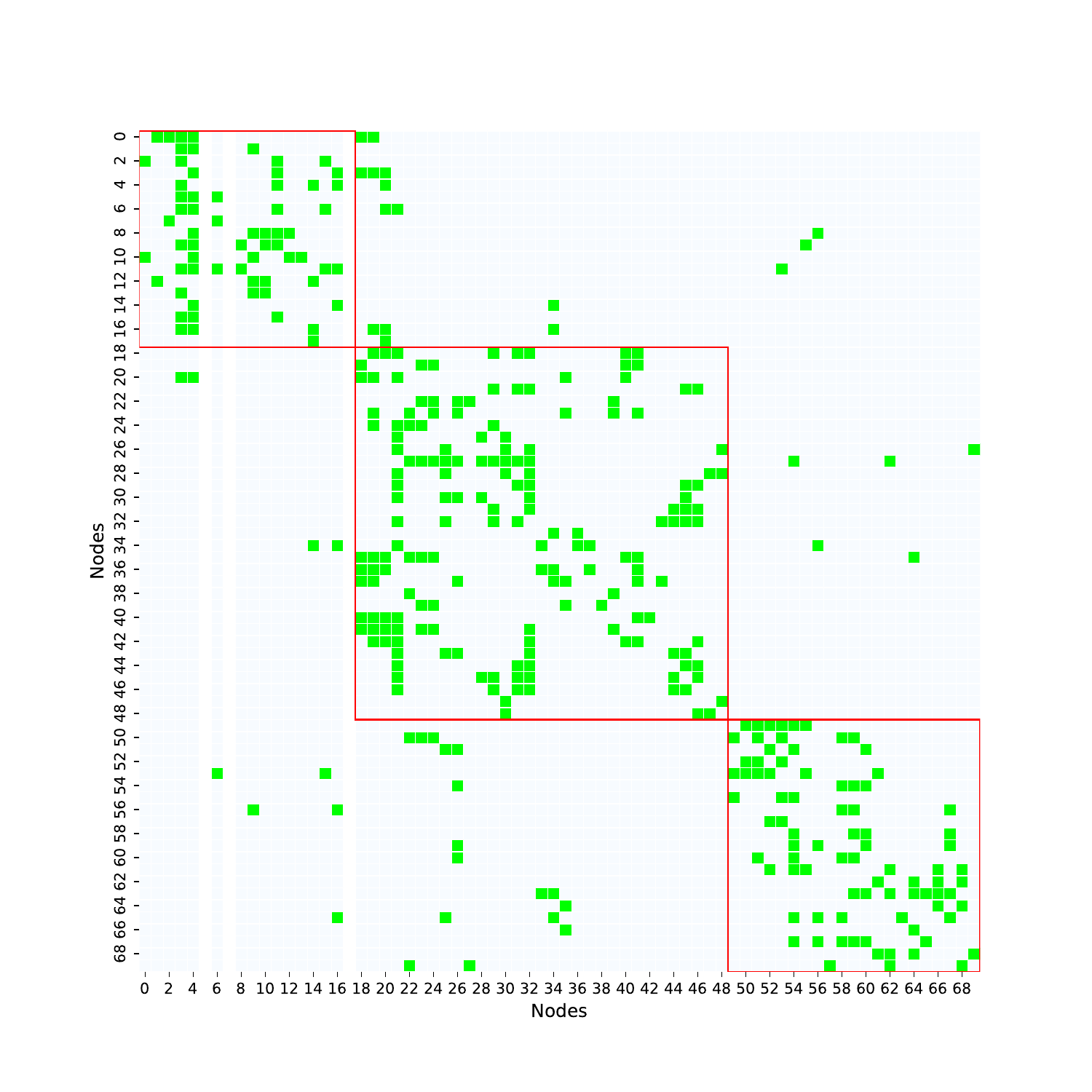}
    \caption{Adjacency matrix associated with Friendship digraph $\G_{\text{f}}$ clustered into three clusters using spectral clustering.}
    \label{HeatmapAdjmatrix}
\end{minipage}
\end{figure}

\begin{table}[tbp]
    \centering
    \caption{Minimum and Maximum Values of $(\delta_{\V}, \Delta_{\V})$ and Condition Number for $w=1$ and $w=0.01$}
    \label{tab:delta_Delta_cond}
    
    \begin{tabular}{ccccc}
        \toprule
        \textbf{Weight} & & $\delta_{\V}$ & $\Delta_{\V}$ & \textbf{Condition Number} \\
        \midrule
        \multirow{2}{*}{$w=1$} & Minimum & $0.0359$  & $0.1707$ & $39.011$ \\
                               & Maximum & $0.1195$  & $0.3165$ & $5.021e+16$ \\
        \midrule
        \multirow{2}{*}{$w=0.01$} & Minimum & $0.00108$ & $0.00395$ & $427.899$ \\
                                  & Maximum & $0.00119$ & $0.00457$ & $3.181e+15$ \\
        \bottomrule
    \end{tabular}
\end{table}

The digraph in Figure \ref{Friendship} shows the friendship connections among boys at a small high school in Illinois, which has 70 students. We will denote this digraph as $\G_{\text{f}}$.
%It is divided into three clusters by the spectral clustering algorithm.
Each boy was surveyed once in the fall of 1957 and again in the spring of 1958. The data set\footnote{https://networks.skewed.de/net/highschool} consolidates the results from the initial survey. In this network, a node represents a boy, and an edge between two nodes indicates that the boy on the left selected the boy on the right as a friend. 
The adjacency matrix is of rank 66, signifying its singular nature. 
Consequently, we must firstly identify and incorporate 4 additional edges into the digraph to successfully dismantle all Jordan blocks linked to the eigenvalue 0. To enhance our visualization, the heat map of the adjacency matrix (Figure \ref{HeatmapAdjmatrix}) is arranged into three distinct clusters using the spectral clustering algorithm. The digraph contains three sinks, which are highlighted in the heat map. Each sink must be incorporated into every column dependency list.
 There are only 5 column dependency lists in total, but there are $15495$ different row dependency lists. Therefore, through this approach, the total number of ways to make the initial adjacency matrix non-singular is as follows:
\[
\text{Total}_{\text{Inv}} = 15,495 \times 5 \times 4! = 1,859,400
\]
By avoiding multiple edges, 596,713 different cases of admissible extensions are emerged.

\begin{remark}
\label{lem}
It is worth mentioning that while the adjacency matrix lacks a source, it does contain two duplicate rows. Interestingly, the exclusion of these duplicates does not impact the rank of the matrix. Additionally, there exist only 1555 row dependency lists, which encompass the two duplicate rows, referred to as D-Rows. The numerical findings suggest that, to attain the desired outcome, it is adequate to limit our consideration to instances where pseudo-permutations are constructed from row dependency lists derived from D-Rows. Within this context, we focus on a total of 7775 non-singular extensions, with 7243 of these cases constituting admissible extensions avoiding multiple edges. Henceforth referred to as the {\it 7243-sample set} in the remainder of this discussion.
\end{remark}

\subsubsection{Analyzing Eigenvector Approximation}
We begin by examining the indices $(\delta,\Delta)$ indices  (see (\ref{F_delta}))  derived from the graph Fourier transforms on the adjacency matrices in the $7243$-sample set, without considering any weights on the newly added edges.
The datas presented in Table \ref{tab:delta_Delta_cond} indicate a high degree of spectral compatibility. However, there is a noticeable variation in the condition numbers of these transforms, which warrants further investigation.
%We denote $\G_{\text{Min-cond}}$ as the digraph corresponding to the minimum condition number among the graph Fourier transforms. 

The distribution of these indices is illustrated in Figures \ref{cond_unweighted}. A similar analysis is presented in Table \ref{tab:delta_Delta_cond}, which focus on the scenario where new edges are introduced with a fixed weight of $w=0.01$.

The high condition numbers observed when the newly added edges are assigned a weight of $0.01$, along with the small indices of $\delta_{\V}$ and $\Delta_{\V}$ when no weight is assigned, serve as compelling evidence in favor of implementing the $7243$-sample set without considering any weight. This decision is based on the analysis of the provided data, which suggests that this approach yields.
%---------------------------------------------------------------

%================================================================
\begin{figure}[tbp]
\centering
\includegraphics[width=0.425\textwidth,height=0.17\textheight]{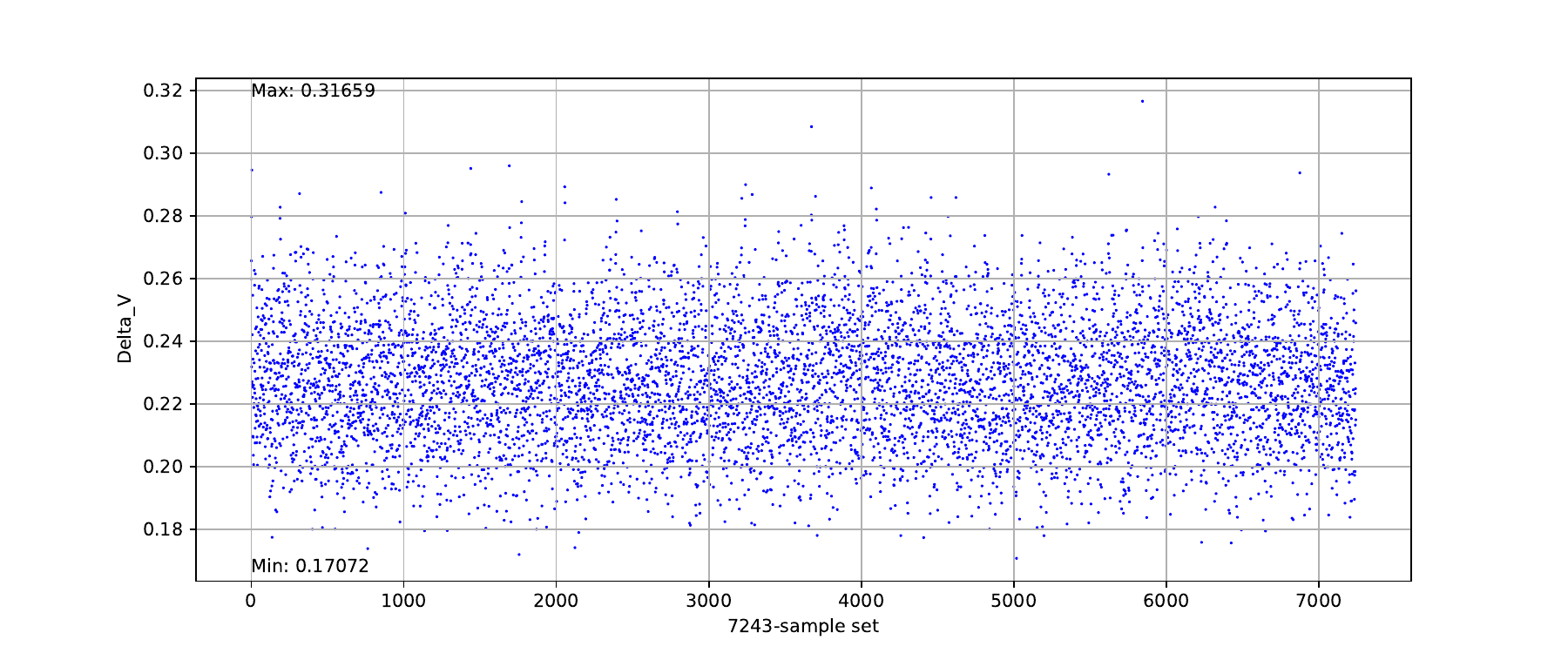}
\\ (a) \\
\includegraphics[width=0.425\textwidth,height=0.17\textheight]{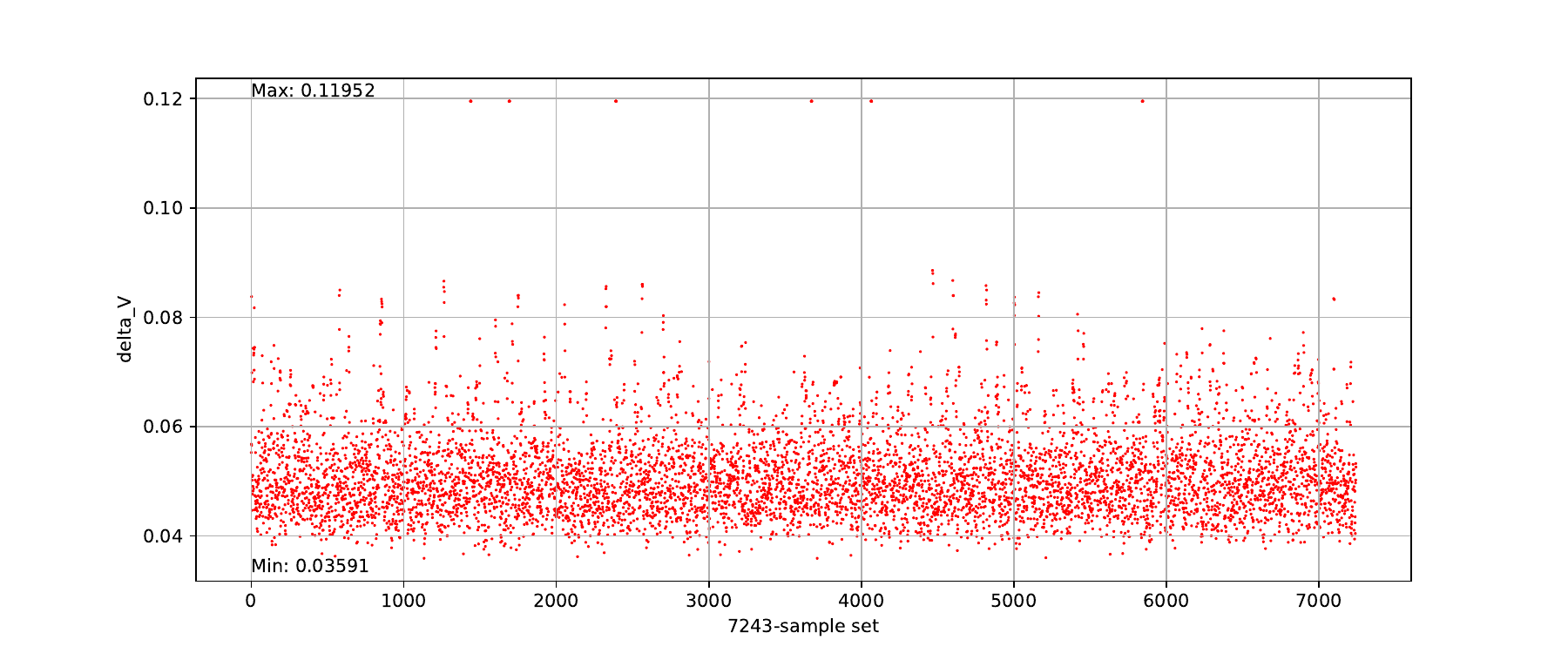}
\\ (b) \\
\includegraphics[width=0.425\textwidth,height=0.17\textheight]{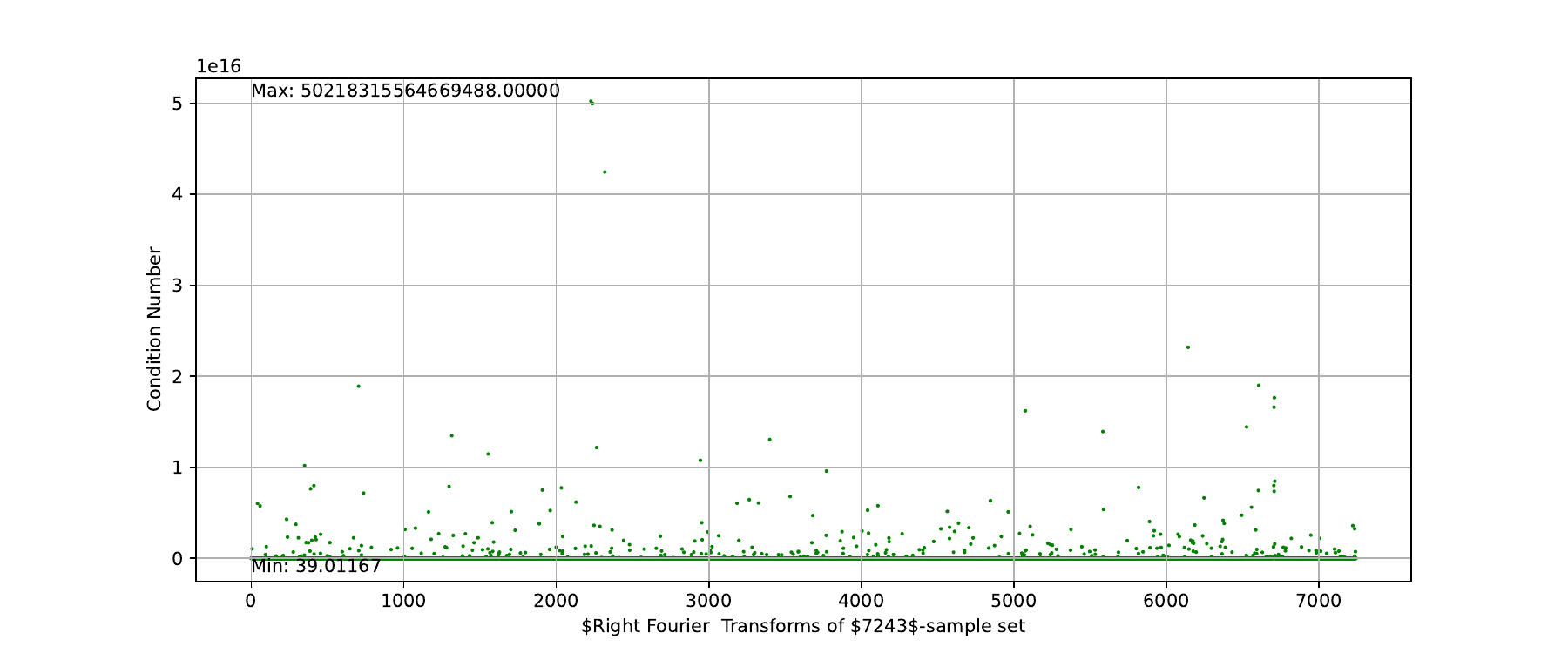}
\\ (c)
\caption{  Figures (a), (b), and (c) illustrate the distribution of $(\delta,\Delta)$ indices and condition numbers of the GFTs obtained from the $7243$-sample set, without taking any weights into account.}
\label{cond_unweighted}
\end{figure}
%---------------------------------------------------------

%---------------------------------------------------------
\subsubsection{Structure Compatibility}
As the structural integrity of the original digraph is increasingly preserved upon the addition of new edges, the graph Fourier transform logically  becomes a more appropriate choice for processing graph signals.

In the subsequent discourse, we propose a framework for determining suitable admissible extensions of the given digraph that prioritize maintaining the structural integrity between the original graph and its extended version.

\begin{remark}
In the process of comparing social networks, numerous indices are frequently employed, including centrality measures such as PageRank, and additional metrics like the Global Clustering Coefficient,  Network Motif Density, and Number of Core/Periphery nodes. By applying these metrics to the $7243$-sample set, a filtering mechanism is established to identify the most appropriate admissible extensions that exhibit the highest structural integrity with the Friendship digraph. This enables the preservation of essential properties and characteristics while extending the network. Brief explanations of these metrics are provided in the Appendix section.
\end{remark}

\begin{figure}[tbp]
\centering
\includegraphics[width=0.425\textwidth,height=0.17\textheight]{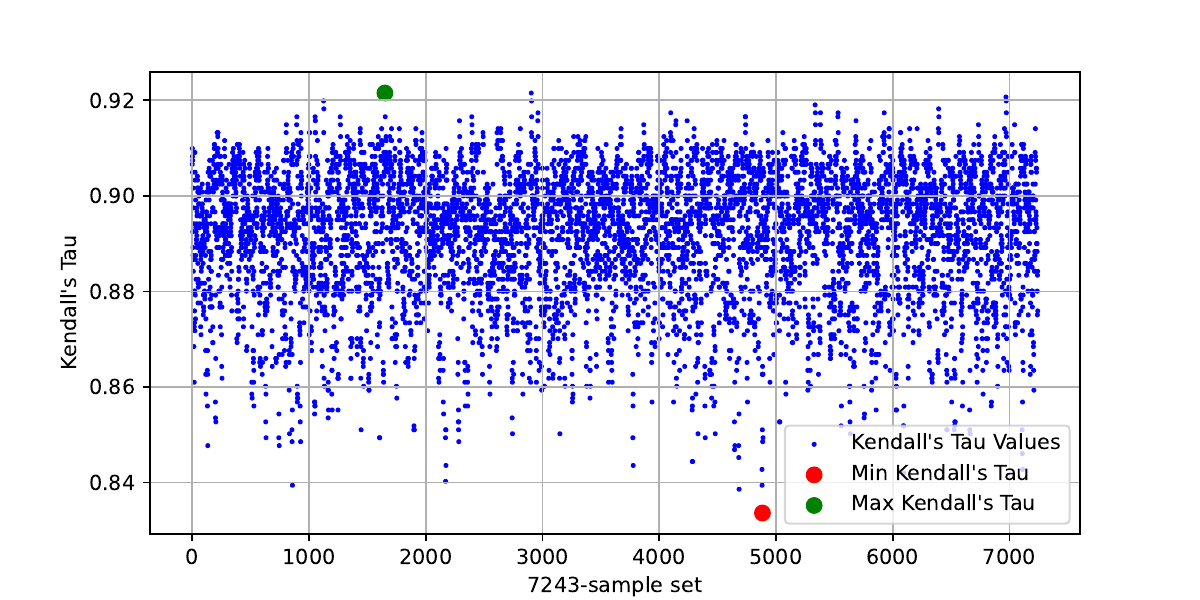}
\caption{ The distribution of Kendall's Tau indices corresponding to the $7243$-sample set reflects the rate of changes in the PageRanks of nodes due to reshuffling after adding new edges, without considering any weight.}
\label{Kendall}
\end{figure} 

{\it Filter 1: PageRank-based Filtering.} To initiate our analysis, we prioritize the PageRank criterion, a fundamental measure of node significance within networks. The Kendall's Tau index informs us about the rate of changes in the PageRanks of nodes that occur due to reshuffling after adding new edges. %Let \(\G_{\text{Max-Tau}}\) denote the digraph with the highest Kendall's Tau index among all admissible digraphs. 
The distribution of Kendall's Tau indices corresponding to the 7243-sample set is depicted in Figure \ref{Kendall}. 
\begin{figure}[tbp]
\centering
\includegraphics[width=0.425\textwidth,height=0.17\textheight]{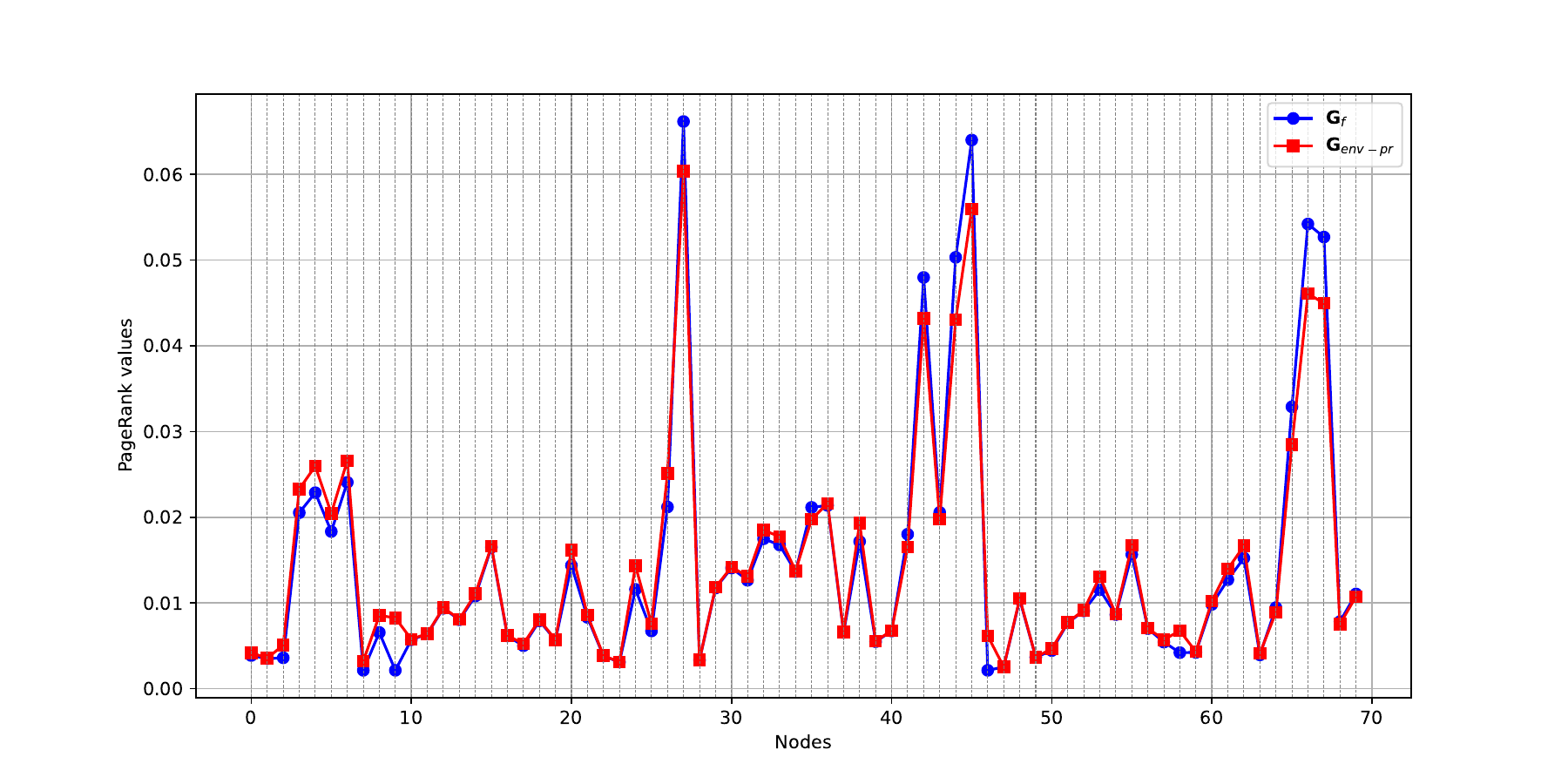}
\caption{  Comparison of PageRanks for digraphs $\G_{\text{f}}$ and $\G_{\text{env-pr}}$.}
\label{PageRank_Comparison}
\end{figure}
%------------------------------------------------
\begin{table}[tbp]
    \centering
    \caption{Within the 7243-sample set, \(\G_{\text{Max-Tau}}\) denotes the digraph exhibiting the highest Kendall's Tau index. Among the list of envelope Friendship-digraphs, \(\G_{\text{env-pr}}\) represents the digraph with the maximum Kendall's Tau index. The condition numbers corresponding to the Graph Fourier Transforms of these digraphs are provided in the third column.
}
    \label{tab:kendall_tau_value}
    \begin{tabular}{ccc}
        \toprule
        \textbf{Indices} & \textbf{Kendall's Tau value} & \textbf{Condition number} \\
        \midrule
        $\G_{\text{Max-Tau}}$ & $0.92148$ & $81062064.706$ \\
%        $\G_{\text{Min-cond}}$ & $0.905738$ & $ 175.59$ \\
        $\G_{\text{env-pr}}$ & $0.914854$ & $69.14729$ \\
        \bottomrule
    \end{tabular} 
\end{table}

We focus on those admissible digraphs that satisfy two specific conditions: a Kendall's Tau index of at least $0.91$ and a condition number of $80$ or less. These criteria ensure a balance between spectral and structural integrity in relation to the original Friendship digraph.

Out of the admissible digraphs in the $7243$-sample set, $223$ possess a Kendall's Tau index of $0.91$ or higher, indicating strong concordance with the original graph's ordering. However, only $7$ of these digraphs meet the condition number threshold of $80$ or less, which is crucial for maintaining numerical stability and ensuring well-posedness in spectral analysis. We refer to these $7$ digraphs as {\it envelope Friendship-digraphs}. 

Within the 7243-sample set, the digraph exhibiting the highest Kendall's Tau index is designated as \(\G_{\text{Max-Tau}}\), while \(\G_{\text{env-pr}}\) represents the digraph that achieves the maximum Kendall's Tau index among all envelope Friendship-digraphs.  In Figure \ref{PageRank_Comparison}, PageRanks of the digraphs $\G_{\text{f}}$ and $\G_{\text{env-pr}}$ are compared together.

 The findings in Table \ref{tab:kendall_tau_value} underscore the importance of achieving a balance between condition number  and structural compatibility when selecting the optimal graph Fourier transform for the Friendship digraph. %This balance ensures that the chosen transform effectively captures both the spectral properties and the structural characteristics of the graph.

{\it Filter 2: Motif Density Filtering.} 
Motif Density assesses the frequency and prevalence of recurring significant patterns or motifs within a graph structure. 
As shown in Figure \ref{fig:motif}, a comparative analysis of the network motif density indices for the envelope Friendship-digraphs and the Friendship-digraph index reveals a striking similarity. In order to facilitate a more comprehensive analysis, an additional digraph, termed \(\G_{\text{env-motif}}\), is introduced and highlighted in the plot using a yellow marker.

\begin{figure}[tbp]
\centering
\includegraphics[width=0.425\textwidth,height=0.17\textheight]{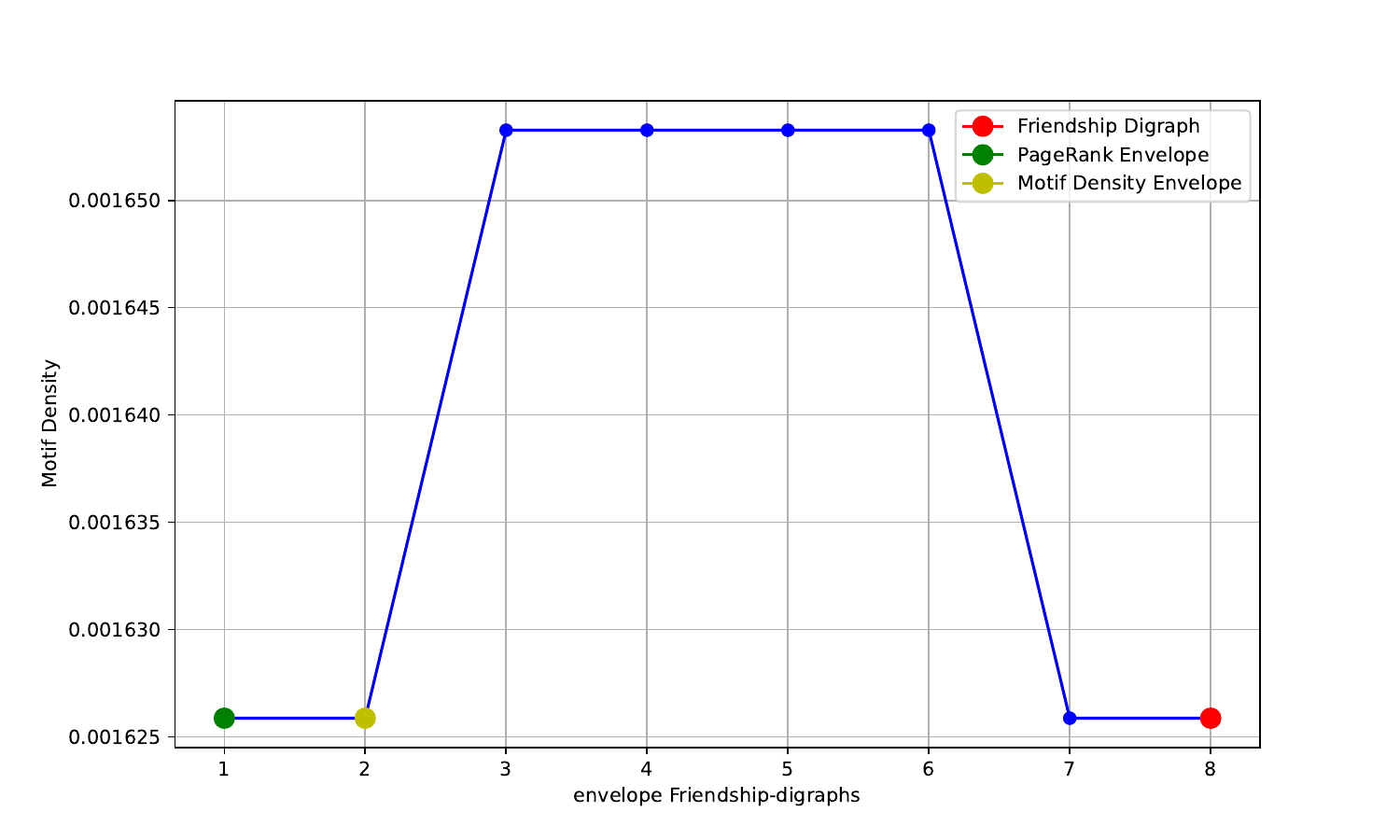}
\caption{The Motif density indices of $\G_{\text{f}}$ and envelope Friendship digraphs.}
\label{fig:motif}
\end{figure}

%---------------------------------------------------
{\it Filter 2: Core and Periphery counts Filtering.} 
Core-Periphery counts evaluate the number of nodes that can be classified as either core or periphery within a graph structure. Core nodes are highly interconnected and central, while periphery nodes are less connected and occupy the outer layers of the network. In Figure \ref{Coreperiphery}, the Core and Periphery counts for the digraphs within the envelope Friendship-digraphs list are provided. The indices for the digraphs and notable digraphs such as $\G_{\text{f}}$, $\G_{\text{env-pr}}$, $\G_{\text{env-motif}}$, and a newly introduced $\G_{\text{env-cor}}$ are highlighted for emphasis.

\begin{figure}[tbp]
\centering
\includegraphics[width=0.425\textwidth,height=0.17\textheight]
{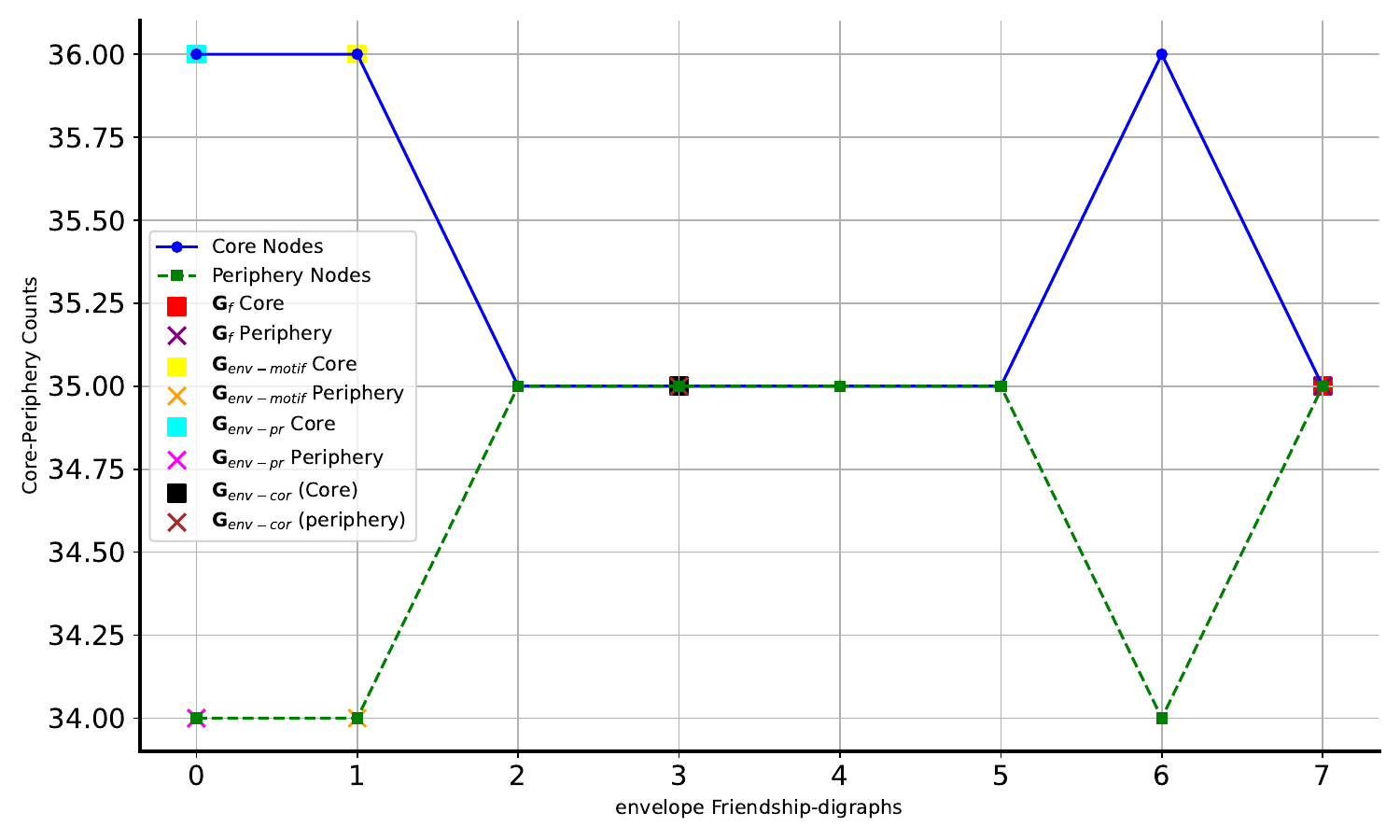}
\caption{Comparison of Core and Periphery Node Counts in $\G_{\text{f}}$ and Envelope Friendship Digraphs}
\label{Coreperiphery}
\end{figure}

%-----------------------------------------------------------
{\it Local clustering coefficients.}   The local clustering coefficient of a node in a graph measures how close its neighbors are to being a complete subgraph (clique). It assesses the extent to which nodes tend to cluster together locally within their immediate neighborhoods.
As demonstrated in Figure \ref{Local_clustering_coefficients}, a comparative analysis of the local clustering coefficients for the notable digraphs $\G_{\text{f}}$, $\G_{\text{env-pr}}$, $\G_{\text{env-motif}}$, and $\G_{\text{env-cor}}$ is conducted, offering valuable insights into their integrity.
\begin{figure}[tbp] 
\centering
\includegraphics[width=0.425\textwidth,height=0.17\textheight]
{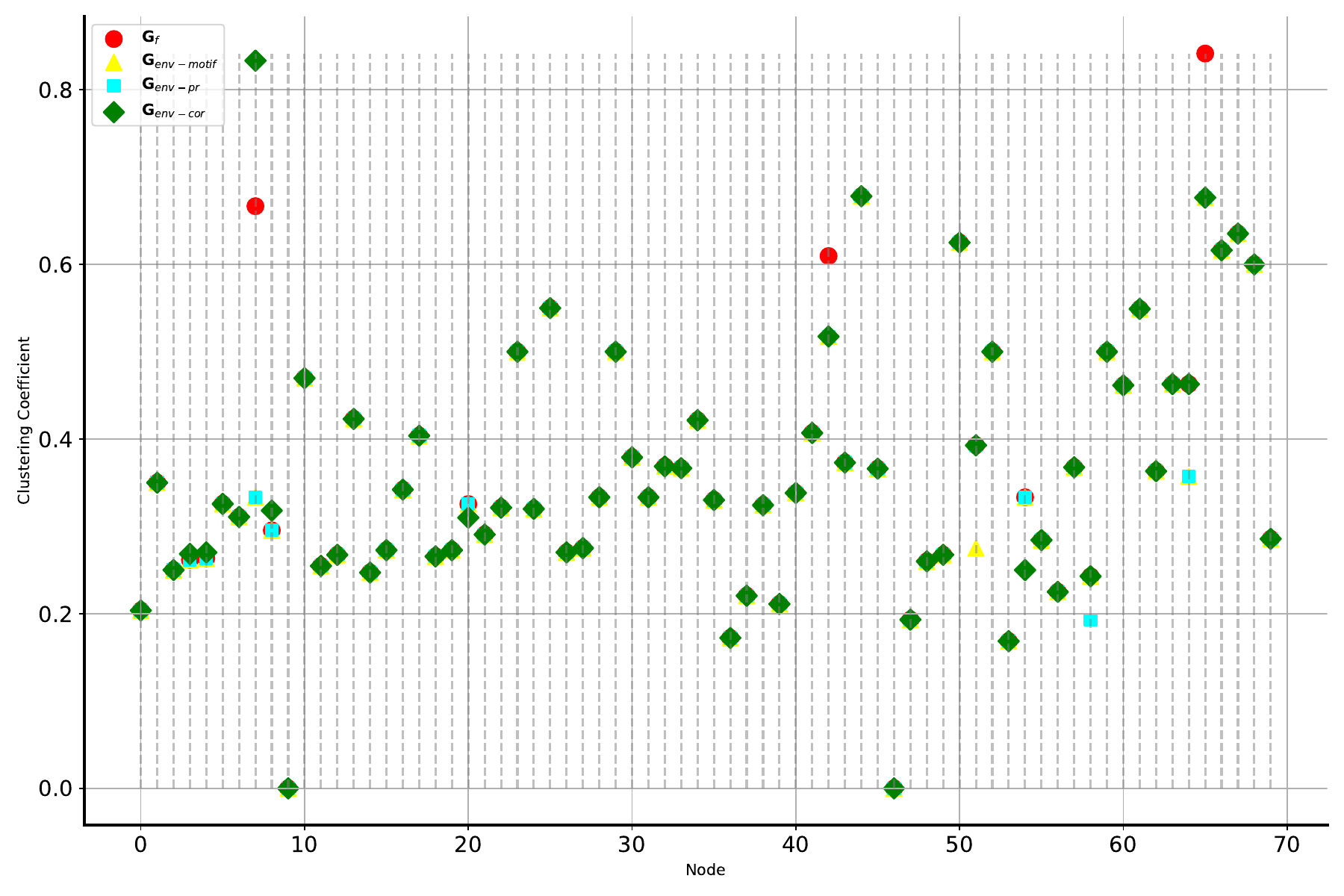}
\caption{Comparing Local Clustering Coefficients in $\G_{\text{f}}$, $\G_{\text{env-pr}}$, $\G_{\text{env-motif}}$, and $\G_{\text{env-cor}}$.}
\label{Local_clustering_coefficients}
\end{figure}
%------------------------------------------------------

{\it Numerical stability :  envelope Friendship-digraphs.}
To assess numerical stability, we compute the operator norms of $(\F \F^{-1} - \mathbf{I})$ and $(\F^{-1} \F - \mathbf{I})$, where $\mathbf{F}$ denotes the GFT of the envelope Friendship-digraphs. This computation evaluates the deviation between the product of inverse GFTs and the identity matrix, providing a thorough analysis of numerical stability. 
The results are illustrated in Figure \ref{LR_RL}.

\begin{figure}[tbp]
\centering
\includegraphics[width=0.425\textwidth,height=0.17\textheight]{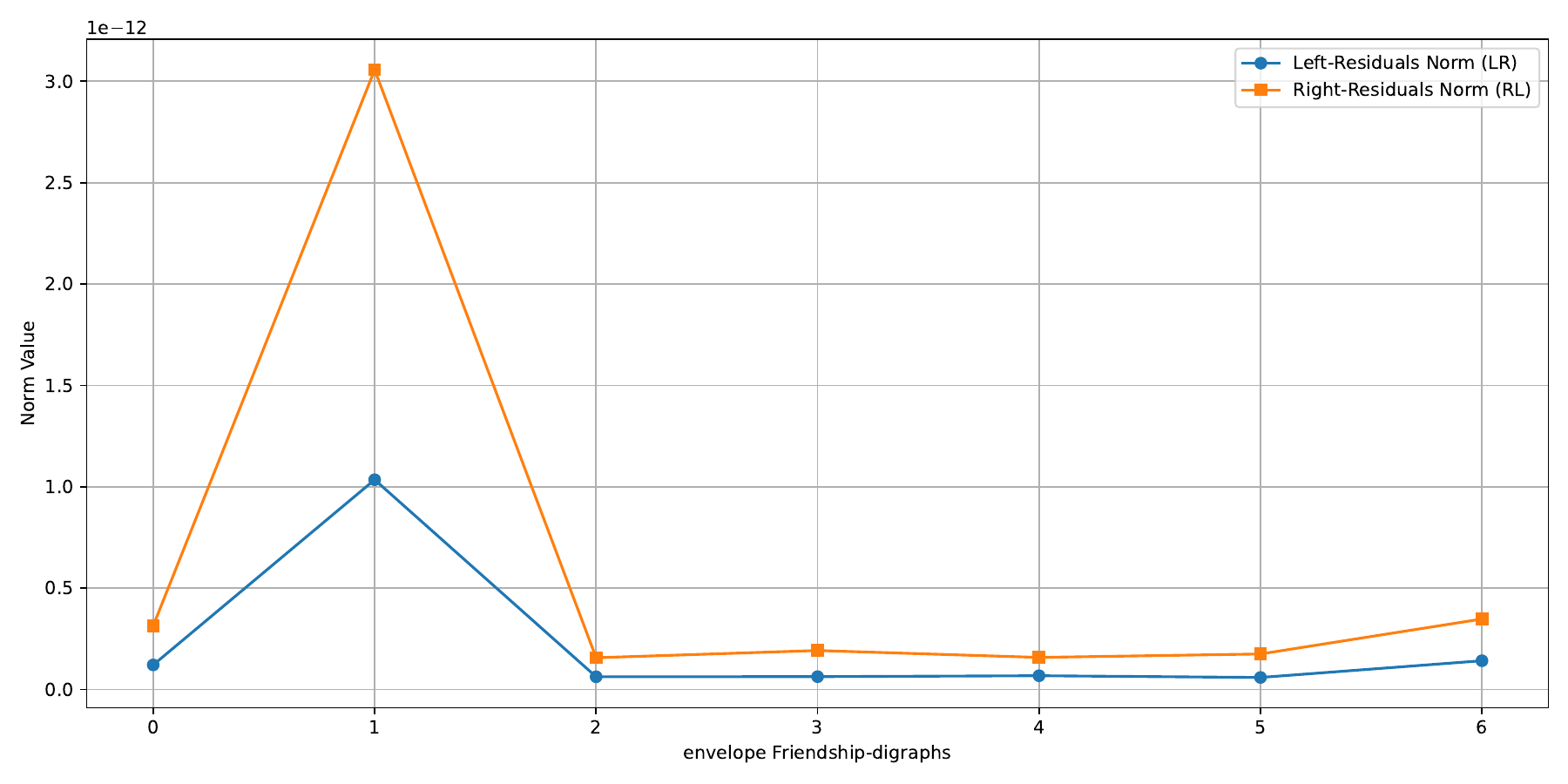}
\caption{This figure illustrates the operator  norms of $(\F \mathbf{F}^{-1} - \mathbf{I})$ and $(\mathbf{F}^{-1} \F - \mathbf{I})$ for GFT of envelope Friendship digraphs.}
\label{LR_RL}
\end{figure}

\subsection{Systems : Target envelopes}
The methodology involves calculating three distinct graph Fourier transforms, referred to as \textit{target envelopes}, specifically $\G_{\text{env-pr}}, \G_{\text{env-cor}}, \G_{\text{env-motif}}$, for the Friendship digraph $\G_{\text{f}}$. These target envelopes introduce the GFTs $\mathbf{F}_{\text{pr}}, \mathbf{F}_{\text{cor}}, \mathbf{F}_{\text{motif}}$ for $\G_{\text{f}}$, respectively. Each of these transforms establishes a specific graph Fourier basis on \(\G_{\text{f}}\). Our subsequent analysis focuses on checking and interpreting the systems induced on $\G_{\text{f}}$ by these target envelopes.

Before delving into this analysis, Figure  (\ref{fig:Cor_heatmap}) depict the heat maps of the corresponding adjacency matrices $\mathbf{A}_{\text{pr}}, \mathbf{A}_{\text{cor}}, \mathbf{A}_{\text{motif}}$ of these three envelope targets, highlighting the newly added edges in each one. The list of newly added edges is also addressed in Table \ref{tab:newly_added}.  The heatmaps presented in Figure \ref{fig:magnitude_comparisons} provide a visual comparison of the spectral characteristics of the graphs for the three target envelopes. These heatmaps illustrate the magnitude differences between the eigenvectors of the adjacency matrices corresponding to the target envelope digraphs, offering a comprehensive representation of their spectral properties and allowing for an intuitive understanding of their similarities.

\begin{table}[tbp]
    \caption{The newly added edges to the Friendship-digraph, leading to the formation of three target envelopes. The notation \((r, c)\) signifies the edge connecting node \( r \) to node \( c \).}
        \centering
    \label{tab:newly_added}
    \begin{tabular}{l|c|c|c|c}
        \toprule
        & \textbf{(r,c)} & \textbf{(r,c)} & \textbf{(r,c)} & \textbf{(r,c)} \\ 
        \midrule
        $\G_{\text{env-pr}}$ & (43,5) & (29,7) & (44,17) & (7,37)  \\ 
        \hline
        $\G_{\text{env-motif}}$ & (43,5) & (29,7) & (44,17) & (7,66)  \\ 
        \hline
        $\G_{\text{env-cor}}$ & (11,5)  & (29,7)  & (44,17)  & (7,68)  \\ 
        \bottomrule
    \end{tabular}
\end{table}
\begin{figure}[htbp]
    \centering
   \includegraphics[width=0.47\textwidth, height=0.4\textheight]{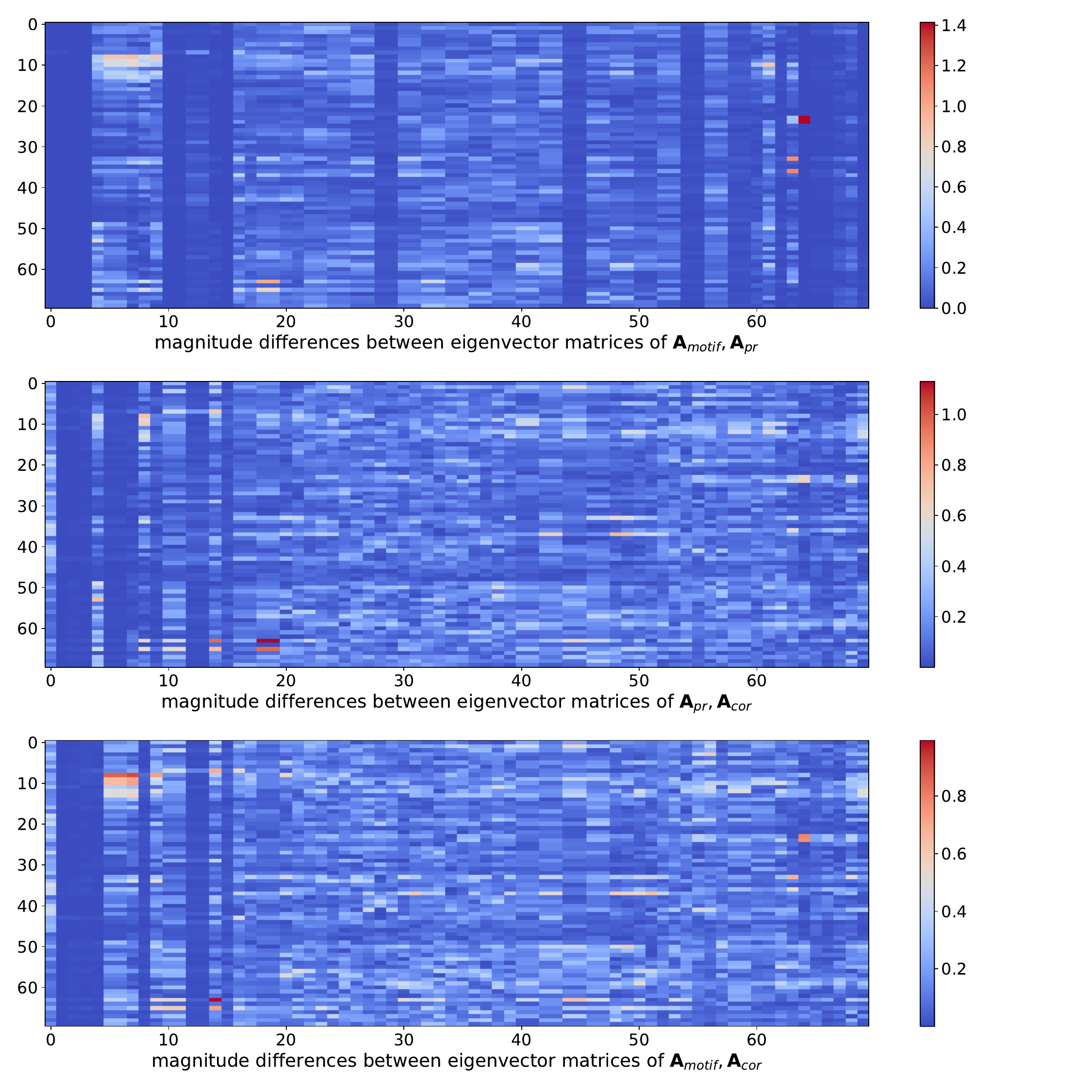}
    \caption{The figures visualize the magnitude differences between columns for any given pair of eigenvector matrices of the adjacency matrices corresponding to the target envelope digraphs.}
    \label{fig:magnitude_comparisons}
\end{figure}

%--------------------------------------------------
\begin{figure}[tbp]
    \centering
       \includegraphics[height=0.29\textheight]{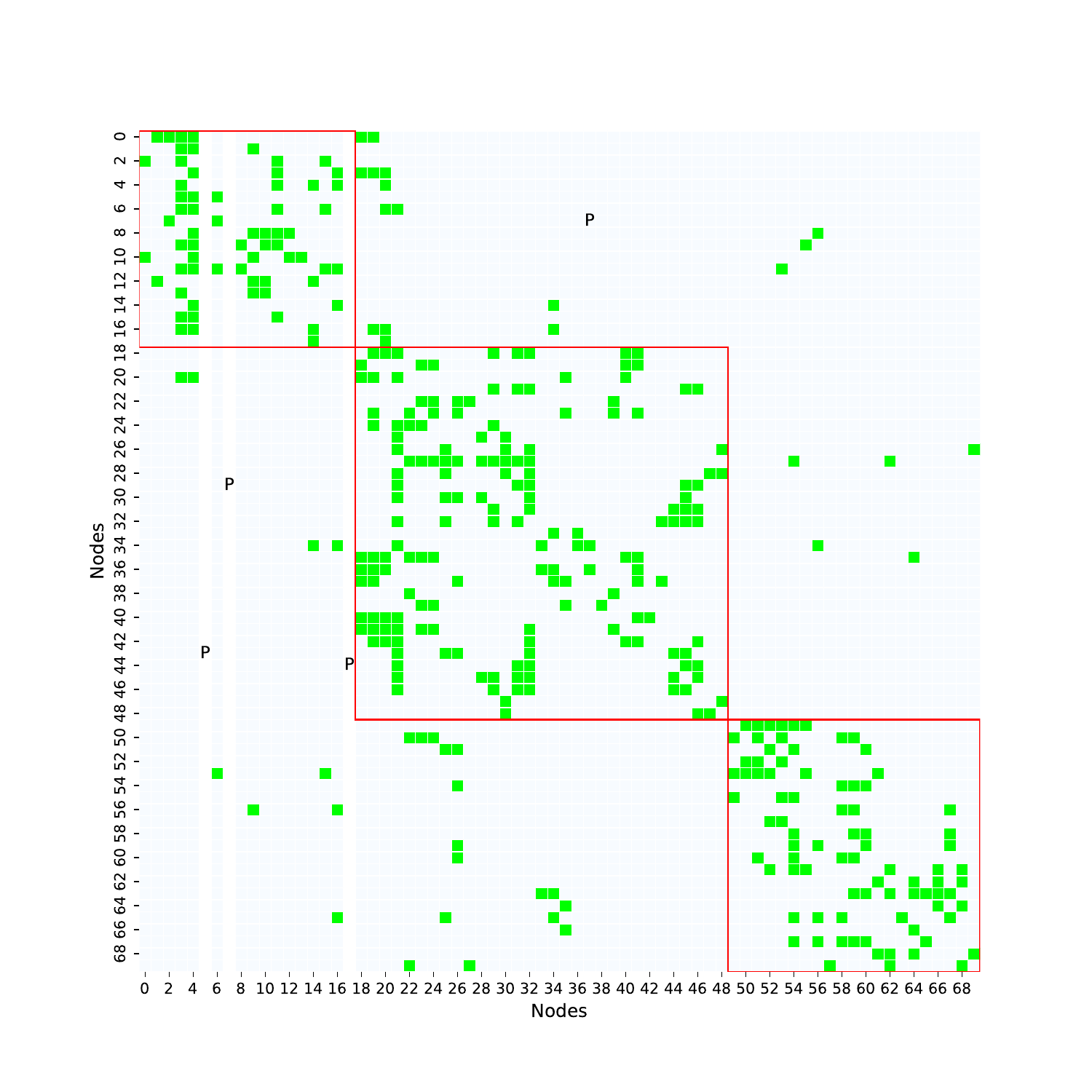} \\ (a) 
       
    %\caption{The heatmap representation of the adjacency matrix for the graph $\G_{\text{env-pr}}$ is shown, where four newly added edges are highlighted with the capital letter P.}
        \label{fig:pagerank_heatmap}
    \includegraphics[height=0.29\textheight]{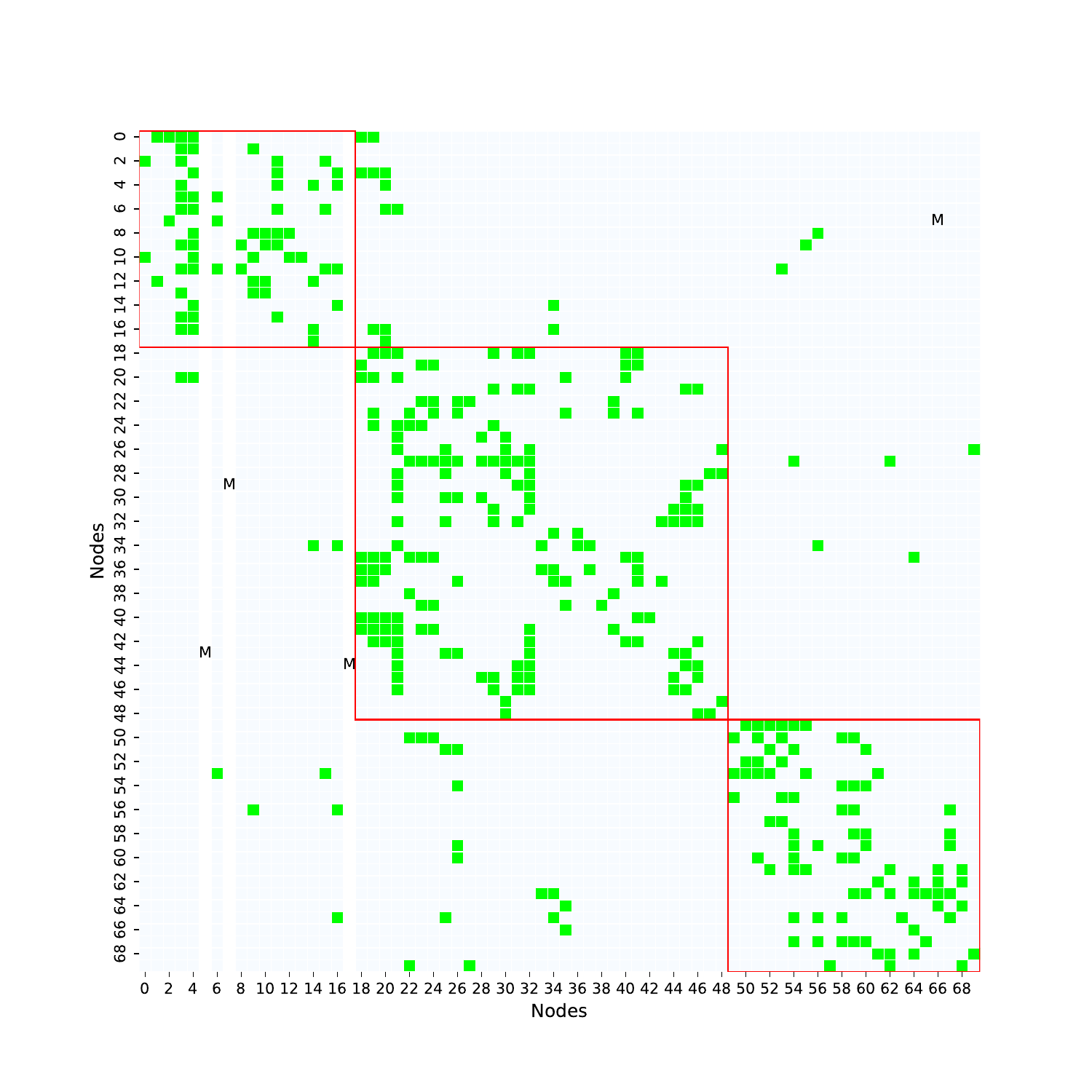}\\(b)\\
   %\caption{The heatmap representation of the adjacency matrix for the graph $\G_{\text{env-motif}}$ is shown, where four newly added edges are highlighted with the capital letter M.}
       \label{fig:Motif_density_heatmap}
   \includegraphics[height=0.29\textheight]{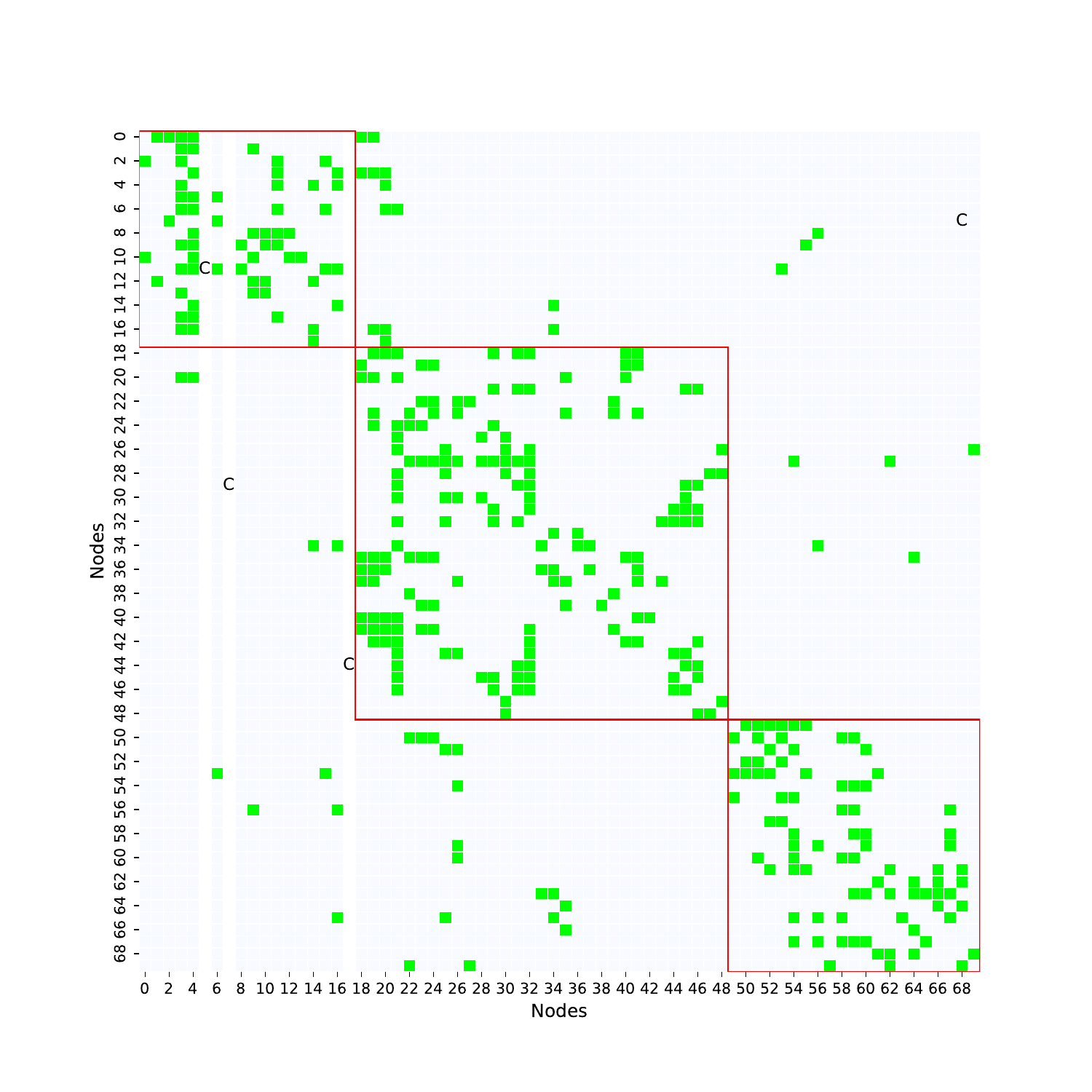}
   \\ (c)
    \caption{The heatmap representation of the adjacency matrix for the digraphs $\G_{\text{env-pr}}$ (a), $\G_{\text{env-motif}}$ (b), $\G_{\text{env-cor}}$ (c) are shown, where four newly added edges are highlighted with the capital letters P, M and C respectively.}
           \label{fig:Cor_heatmap}
\end{figure}
%------------------------------------------------------

\begin{figure*}[htbp]
    \centering
\includegraphics[width=0.75\textwidth, height=0.55\textheight]{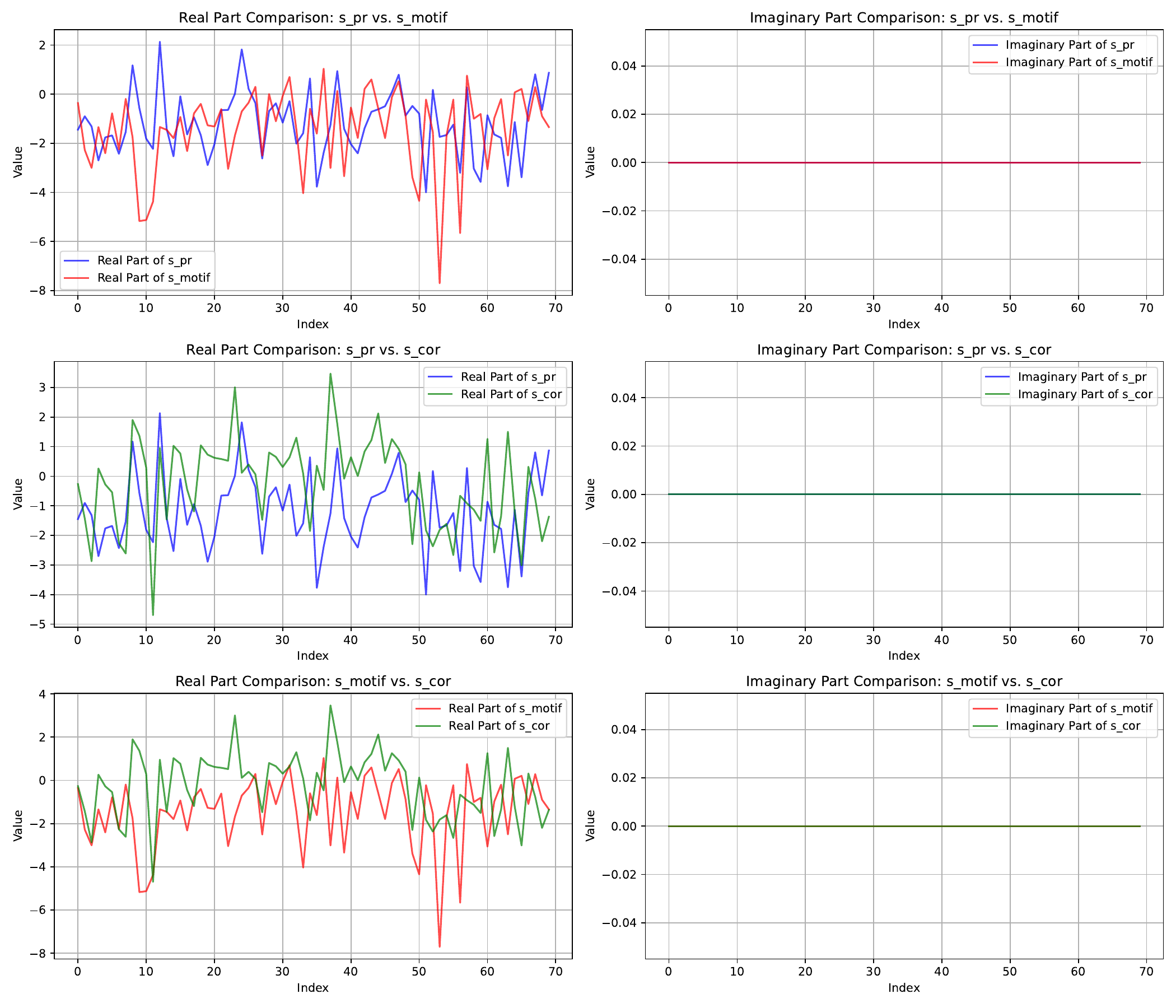}

    \caption{The real and imaginary parts of signals defined in (\ref{con-impulses}) are compared together in the figures.}
    \label{fig:systems}
\end{figure*}

The following graph signals illustrate the systems induced by the target envelopes.
\begin{equation}
\label{con-impulses}
\begin{cases}
\mathbf{s}_{\text{pr}} = \F_{\text{pr}}^{-1}\D_{\text{pr}}\mathbf{1} \\
\mathbf{s}_{\text{cor}} = \F_{\text{cor}}^{-1}\D_{\text{cor}}\mathbf{1} \\
\mathbf{s}_{\text{motif}} = \F_{\text{motif}}^{-1}\D_{\text{motif}}\mathbf{1} 
\end{cases}
\end{equation}
where $\mathbf{D}_{\text{pr}}$, $\mathbf{D}_{\text{cor}}$, and $\mathbf{D}_{\text{motif}}$ are diagonal matrices containing the eigenvalues of $\mathbf{A}_{\text{pr}}$, $\mathbf{A}_{\text{cor}}$, and $\mathbf{A}_{\text{motif}}$, respectively. 
Figures \ref{fig:systems} compare the real and imaginary parts of the signals defined in (\ref{con-impulses}) . 

{\bf Summary.} In summary, the findings of this section are as follows:

i) A filter mechanism was developed to identify appropriate GFTs for the digraph $\G_{\text{f}}$, focusing on achieving maximal spectral and structural compatibility. Consequently, three target envelopes were obtained, and it was observed that the resultant digraphs converged to similar characteristics as the filtering process progressed. The adjacency matrices corresponding to $\G_{\text{env-pr}}$ and $\G_{\text{env-motif}}$ were found to be identical, except for the addition of a single new edge with the same source.

ii) To analyze the impact of these envelopes on the initial digraph $\G_{\text{f}}$, comparing shift-invariant graph filters offers a suitable approach. The systems defined by them on $\G_{\text{f}}$ are characterized by their convolution operators, as established in Theorem \ref{con-thm}.

As illustrated in Figure \ref{fig:systems}, the systems induced by these three convolution operators exhibit sufficiently distinct characteristics. This indicates that even with a certain level of spectral and structural  compatibility, each GFT derived from the envelope digraphs corresponds to a unique case. Consequently, the different GFTs demonstrate distinct behaviors, thereby validating the hypothesis that these GFTs have unique impacts on signal processing within graph-based data frameworks.

%-------------------------------------------------
\subsection{Conclusion} 
When a real dataset is represented as a digraph, graph signal processing techniques might be employed for signal analysis. One significant contribution is formulating the challenges to obtain a well-structured graph Fourier basis, which is intricately linked to the nature of the graph structure. To achieve the desired target, modifying the graph structure to develop the GFT could be a viable solution.

Admissibility is a minimal property in a digraph that allows the utilization of shift-invariant graph filters for signal analysis, making them highly effective tools.

For a given non-admissible digraph \(\G\), a systematic approach is delineated to identify a class of sophisticated admissible Cayley digraphs \(\text{Cay}(\mathbb{Z}_N, \Gamma)\). This ensures that \(\G\) fits within various admissible envelopes, enabling the application of graph signal processing techniques to \(\G\).

To achieve spectral and structural compatibility with $\G$, it is necessary to design an appropriate filtering mechanism. The initial problem considered on the real dataset dictates the selection of the optimal option among the admissible envelopes of $\G$ that emerge through the filtering process.

Any filtering mechanism will lead to a specific envelope digraph that is suitable only for that particular case. Consequently, choosing admissible extensions randomly may not yield reliable results.

\section{Appendix}
\subsection{Cayley digraphs  \(\text{Cay}(\mathbb{Z}_N, \Gamma)\)}

Let \( \mathbb{Z}_N \) denote the cyclic group of integers modulo \( N \), where \( N \) is a positive integer. The group \( \mathbb{Z}_N \) consists of the elements \(\{ 0, 1, 2, \ldots, N-1 \}\) under addition modulo \( N \).

Given a subset \( \Gamma \subseteq \mathbb{Z}_N \), the {\it Cayley digraph} \(\text{Cay}(\mathbb{Z}_N, \Gamma)\) is a directed graph defined as follows:

\begin{itemize}
    \item The vertex set of the digraph is the set of elements of \(\mathbb{Z}_N\).
    \item There is a directed edge from vertex \( m \) to vertex \( n \) if and only if \( n - m \in 
    \Gamma\) (mod \( N \)).
\end{itemize}
In other words, the adjacency relation in the Cayley digraph \(\text{Cay}(\mathbb{Z}_N, \Gamma)\) is given by $(m, n)\in \mathbb{Z}_N \times \mathbb{Z}_N$ iff $n \equiv m + k \pmod{N}$  for some $k\in \Gamma$.

The set \( \Gamma \) is often referred to as the {\it connection set} of the Cayley digraph. It is important to note that \( \Gamma \) should not contain the identity element (0 in the case of \(\mathbb{Z}_N\)) to avoid loops at each vertex unless loops are desired in the digraph.

In the specific case $\Gamma=\{1\}$, the Cayley digraph \(\text{Cay}(\mathbb{Z}_N, \Gamma)\) simplifies to the directed cycle digraph with $N$ vertices whose adjacency matrix is just the  circulant  matrix $\mathbf{C}$. It is known  diagonalizable with the following eigenvalue decomposition. 
 \begin{equation}
 \label{circulant}
 \begin{cases}\mathbf{C}=\FF^{-1}\D\mathcal\FF \\
 \D=\text{diag}\Big(1,\exp{(-\frac{2\pi i}{N})},\ldots,\exp{(-\frac{2\pi i(N-1)}{N})}\Big)  
 \end{cases}    
 \end{equation}
Let  $\Gamma=\{n_1,\ldots, n_k\}$  be a connection set and $\A_{\Gamma}$ as the associative adjacency matrix for the Cayley graph \(\text{Cay}(\mathbb{Z}_N, \Gamma)\). Then one may directly check that,  
\begin{align*}
\A_\Gamma &= \sum_{q=1}^{k}\C^{n_q} = \sum_{q=1}^{k}(\FF^{-1}\D\FF)^{n_q} \\
&= \FF^{-1} \Big( \sum_{q=1}^{k} \D^{n_q} \Big) \FF
\end{align*}

The matrix $\D_{\Gamma} = \sum_{q=1}^{k} \D^{n_q}$ is a diagonal matrix whose diagonal entries are precisely the eigenvalues of $\A_{\Gamma}$. Consequently, $\A_{\Gamma} = \FF^{-1} \D_{\Gamma} \FF$ represents the eigenvalue decomposition of $\A_{\Gamma}$. This proves that the discrete Fourier matrix is the graph Fourier transform for all Cayley digraphs \(\text{Cay}(\mathbb{Z}_N, \Gamma)\).

\subsection{A brief description for metrics}
\subsubsection{PageRank Centrality} The PageRank algorithm, initially created by Google for ranking web pages, is used in network analysis to measure node importance based on incoming link structure. Higher scores indicate more central nodes. The algorithm distributes a probability value across the network, considering both the quantity and quality of incoming connections.
\[
PR(i) = \frac{1 - d}{N} + d \sum_{j \in M(i)} \frac{PR(j)}{L(j)}
\]
where:
\begin{itemize}
    \item \( PR(i) \) is the PageRank of node \( i \).
    \item \( d \) is the damping factor, usually set to 0.85.
    \item \( N \) is the total number of nodes in the graph.
    \item \( M(i) \) is the set of nodes that link to node \( i \).
    \item \( L(j) \) is the number of outbound links from node \( j \).
\end{itemize}
The Kendall rank correlation coefficient (Kendall's tau) measures similarity between two rankings. It evaluates ordinal association by calculating differences in item pair orderings.

The formula for Kendall's tau is:
\[{\small \tau = \frac{(\text{number of concordant pairs}) - (\text{number of discordant pairs})}{\binom{n}{2}} }\]

where:
\begin{itemize}
    \item \( \tau \) is Kendall's tau coefficient,
\item  \( n \) is the number of paired observations,
\item \( \binom{n}{2} \) represents the number of pairs of items.
\end{itemize}\medskip

\subsubsection{Motif Density} Motif Density quantifies the prevalence of recurring significant patterns or motifs within a graph structure. The formula for Motif Density can be expressed as:
\[
\text{Motif Density}(M) = \frac{n_M}{N_G}
\]
where:
\begin{itemize}
    \item \( n_M \) is the number of occurrences of motif \( M \) in the graph \( G \).
    \item \( N_G \) is the total number of potential occurrences of motif \( M \) in graph \( G \), considering its size and configuration.
\end{itemize}

This metric helps in understanding the structural importance and frequency of specific patterns within a network, providing insights into network dynamics and characteristics.

\subsubsection{ Core-Periphery Criterion:}
Core-Periphery counts evaluate how nodes within a graph can be categorized into core and periphery groups based on their connectivity patterns.
\begin{itemize}
    \item \textbf{Core Nodes}: These nodes have higher degrees of connectivity and play central roles within the network.
    \item \textbf{Periphery Nodes}: These nodes have fewer connections and are typically situated on the edges of the network.
\end{itemize}
The core-periphery structure score \( S_i \) for each node \( i \) is given by:
\[
S_i = \frac{k_i - \langle k \rangle}{k_i + \langle k \rangle}
\]
where:
\begin{itemize}
    \item \( k_i \) is the degree (number of connections) of node \( i \),
    \item \( \langle k \rangle \) is the average degree of all nodes in the graph.
\end{itemize}
Nodes with \( S_i > 0 \) are classified as core nodes, while nodes with \( S_i \leq 0 \) are classified as periphery nodes.

\subsubsection{Global Clustering Coefficient} The local clustering coefficient of a node in a graph measures the degree to which its neighbors are interconnected. It quantifies how close the node's neighbors are to forming a complete subgraph (clique), thereby indicating the local density of connections around the node.
The local clustering coefficient \( C_i \) for a node \( i \) is given by:
\[
C_i = \frac{2 \times e_i}{k_i \times (k_i - 1)}
\]
where:
\begin{itemize}
    \item \( k_i \) is the degree of node \( i \), i.e., the number of connections it has.
    \item \( e_i \) is the number of edges between the neighbors of node \( i \).
\end{itemize}

This formula computes the ratio of the actual number of edges between the neighbors of node \( i \) (denoted by \( e_i \)) to the maximum possible number of edges that could exist between them (given by \( k_i \times (k_i - 1) / 2 \), which is the number of edges in a complete graph of \( k_i \) nodes).

\bibliographystyle{IEEEtran}

\bibliography{ref.bib}

\end{document}